\documentclass[12pt,a4paper,final,single-columns]{iopart}

\usepackage{iopams}  
\usepackage{graphicx}
\usepackage[breaklinks=true,colorlinks=true,linkcolor=blue,urlcolor=blue,citecolor=blue]{hyperref}

\usepackage{amsfonts,amssymb,amsthm,array}

\usepackage{verbatim}

\usepackage{mathrsfs}

\usepackage{color}

\def\red{\color{red}}

\def\blue{\color{blue}}
\def\cyan{\color{cyan}}

\theoremstyle{plain}
\newtheorem{thm}{Theorem}[section] 

\newtheorem{prop}{Proposition}[section] 

\theoremstyle{definition}
\newtheorem{defn}[thm]{Definition} 

\newcommand{\eq}[1]{(\ref{#1})}

\begin{document}


\title[title]{Central charge from adiabatic transport of cusp singularities in the quantum Hall effect}

\author{Tankut Can} 
\address{Simons Center for Geometry and Physics, Stony Brook University, Stony Brook, NY 11794, USA}
\ead{tcan@scgp.stonybrook.edu}



\begin{abstract}

We study quantum Hall (QH) states on a punctured Riemann sphere. We compute the Berry curvature under adiabatic motion in the moduli space in the large $N$ limit. The Berry curvature is shown to be finite in the large $N$ limit and controlled by the conformal dimension of the cusp singularity, a local property of the mean density. Utilizing exact sum rules obtained from a Ward identity, we show that for the Laughlin wave function, the dimension of a cusp singularity is given by the central charge, a robust geometric response coefficient in the QHE. Thus, adiabatic transport of curvature singularities can be used to determine the central charge of QH states. 
We also consider the effects of threaded fluxes and spin-deformed wave functions. Finally, we give a closed expression for all moments of the mean density in the integer QH state on a punctured disk. 


\end{abstract}

\pacs{73.43.-f, 02.40.-k, 11.25.Hf, 03.65.Vf}
\vspace{2pc}
\noindent{\it Keywords}: Adiabatic transport, fractional quantum Hall effect, moduli space, K\"ahler metrics\\ 
\submitto{\JPA}

\tableofcontents

\section{Introduction}
Since its discovery, the fractional quantum Hall effect (FQHE) \cite{Tsui1982} has served as a lamppost for exotic quantum phase of matter. From the fractional charge \cite{Laughlin1983} and statistics \cite{Arovas1984,Halperin1984} of its quasi-particle excitations, to the topological nature of its transport coefficients \cite{Niu1985}, the QHE has inspired a broad search for condensed matter systems exhibiting anyonic excitations and topological characterization. 

The QHE continues to be a wellspring of novel insights into the nature of topological states. In particular, it has recently emerged that a complete description of the QHE, and likely topological states in general, requires understanding how these states couple to the background space-time geometry \cite{Gromov2014b,Gromov2014,Abanov2014,Klevtsov2013,Klevtsov2015,Ferrari2014,Can2014,Can2015,Bradlyn2015,Read2011,Son2013,Cho2014,Klevtsov2016a}. This paper is a drop in this stream.

\subsection{Geometric response in the quantum Hall effect: a brief history}

Geometric response is not an obvious thing to study, and does not reveal itself readily in QH systems. 
First of all, the setting of nontrivial spacetime geometry is rather arcane in the context of condensed matter physics, which deals with earthbound materials for which spacetime is strictly flat. Secondly, any effect due to the geometry is generically overwhelmed by the electromagnetic properties. Heuristically, this is due to the competition of length scales, with the magnetic length $l_{B}$ typically setting the ultraviolet cut-off, and geometric effects appearing in powers of the local radius of curvature $r_{c}$, and suppressed by the small parameter $l_{B}/r_{c}$. Correlation functions in the ground state will generically have a long-wavelength expansion in terms of this small parameter. Thus, the effects of spatial curvature will be suppressed in the long-wavelength limit, which is the relevant regime to consider if we are interested in physics below the gap. Therefore, to learn about geometric response, we must pay attention to subleading effects.

The DC electromagnetic response of a QH state is determined by the filling fraction $\nu$. It was shown in Refs. \cite{Wen1992b, Frohlich1992} that on a curved surface $\Sigma$ the local electronic density is (setting $e = \hbar = 1$)

\begin{equation}\label{rhobar}
	\bar{\rho} = \frac{\nu B}{2\pi } + \frac{\mu_{H}}{4\pi} R,
\end{equation}
which introduces an additional coefficient $\mu_{H}$ that is robust on a QH plateau. Here, $B$ is the magnetic field, $R$ is the scalar curvature. The Hall conductance follows from the St\v{r}eda formula $\sigma_{H} =  \partial \bar{\rho} / \partial B = \nu / 2\pi$. 

The geometric term is subleading for large magnetic field, and introduces a kinetic coefficient for geometric transport $\mu_{H}$. In \cite{Read2011}, it was shown that the odd or Hall viscosity $\eta_{A} =  \bar{\rho} \mu_{H}/2\nu $. In \cite{Can2015}, we argued that the odd viscosity is related to the St\v{r}eda formula for the curvature response $\eta_{A} =(2\pi  \bar{\rho}/\nu)\,  \left(\partial \bar{\rho}/ \partial R\right)_{R=0} $. In the literature, $\mu_{H}$ is often referred to as the orbital spin density and denoted by $\bar{s}$. It is also clearly related to the ``shift" $\mathcal{S} = 2 \mu_{H}/\nu$ introduced in \cite{Wen1992b}, since integrating the expression \eq{rhobar} yields $N = \nu N_{\phi} + \mu_{H} \chi(\Sigma)$, where $N_{\phi}$ is the total flux in units of the flux quantum, and $\chi(\Sigma)$ is the Euler characteristic of $\Sigma$. 

The density \eq{rhobar} receives gradient corrections if the curvature is not constant. These were found for the IQH state at $\nu = 1$ and $s = 0$ in \cite{Douglas2010}.  In \cite{Can2015}, the leading correction to the density of the spin-s Laughlin state was found to be 
\begin{equation}\label{rho_grads}
	\langle \rho \rangle = \bar{\rho} + \left( \frac{c_{H}}{12} + \mu_{H}\right) \frac{1}{8\pi}l_{B}^{2} \Delta R + O(l^{4}), \,\,\, c_{H} = 1 - 12  \mu_{H}^{2}/\nu,\,\,
\end{equation}
where for the Laughlin state $\mu_{H} = (1 - 2\nu s)/2$. This correction introduces yet another geometric transport coefficient, which we refer to simply as the central charge $c_{H}$  \footnote{Our central charge $c_{H}$ is equivalent to the``apparent central charge" $c_{app}$ in \cite{Bradlyn2015}}. In \cite{Gromov2014}, local Galilean invariance was shown to imply that $c = c_{H} + 12  \mu_{H}^{2}/\nu$ is the gravitational anomaly or ``chiral central charge" which governs the thermal Hall response\footnote{In \cite{Gromov2014b,Bradlyn2015} it is pointed out that this relation between $c_{H}$ and $c$ can only true when the ``orbital spin variance" vanishes. For instance, it will fail for Jain states. A natural conjecture for $c_H$ (which is a property of a wave function) is that it is generalized to read $c_{H} = c - 12 \nu^{s}$, where $\nu^{s}$ is the orbital spin filling fraction introduced in \cite{Wen1992b}. $\nu^{s} = \nu^{-1} \mu_{H}^{2}$ only for vanishing orbital spin variance.}. The connection is indirect. Local Galilean invariance implies a Ward identity which connects the coefficient $c_{H}$ appearing in the mean density to the coefficient of the gravitational Chern-Simons term in the induced action. The latter is believed to control the thermal Hall coefficient, which is proportional to the chiral central charge \cite{Read2000, Stone2012, Gromov2016a, Gromov2014b}. For example, $c=1$ for the Laughlin and $\nu = 1$ integer QH state, $c = \nu$ for integer $\nu$-filled state. Furthermore, the odd viscosity was shown to have a finite size correction on a sphere $\eta_{A}   =  \mu_{H} B / 4\pi  - (c_{H}/24) R  / 4\pi $ \cite{Gromov2014b}. 

\paragraph{ Adiabatic Transport} Adiabatic transport and Berry curvature provide another window into the geometric response of QH states. The seminal work of Avron, Seiler and Zograf \cite{Avron1995} introduced the odd (Hall) viscosity and demonstrated that it was given by the Berry curvature in the moduli space of a torus. For the $\nu = 1$ IQH state, they computed the Berry curvature

\begin{equation}
 \Omega = -\frac{N_{\phi}}{	4} \frac{i\,  d\tau \wedge d \bar{\tau}}{2 {\rm Im}(\tau)^{2}} ,
\end{equation}
where $\tau$ is the complex structure moduli with ${\rm Im}(\tau) >0$. This result was also obtained by L\'evay in \cite{Levay1995}, who generalized it to $\nu$ filled Landau levels. Importantly, it was observed in \cite{Levay1995} that for the lowest Landau level (LLL) the Berry curvature was a K\"ahler form, and that the normalization integral of the many-electron wave function was the K\"ahler potential. In Ref. \cite{Levay1997a}, this result was further generalized to higher genus surfaces. There, a connection was made between the normalization integral and the Quillen metric, and the K\"ahler property was used to compute the Berry curvature 

\begin{equation}\label{BerryLevay}
	\Omega = - \frac{1}{12\pi} \left( 6B^{2} -6B + 1\right) \omega_{WP} + i\partial \bar{\partial} \log Z (B),
\end{equation}
where $Z(B)$ is the Selberg zeta function, and $\omega_{WP}$ is the Weil-Petersson metric on the moduli space (see Sec. \ref{modulispace}). 

The crucial observation that the Berry curvature is K\"ahler was used in \cite{Klevtsov2015} to extend L\'evay's result to the Laughlin series of fractional  quantum Hall states with integer inverse filling $\nu^{-1}$. Their result implies that for FQH states on surfaces of constant curvature $R_{0}$, \eq{BerryLevay} is generalized to read 

\begin{eqnarray}\label{BerryLevayFQH}
	\Omega
	&= \frac{1 }{\pi } \left( \nu \frac{B^{2}}{R_{0}}  + \mu_{H} B -  \frac{c_{H}}{48} R_{0} \right) \omega_{WP} + \Omega_{E}
\end{eqnarray}
where the $\Omega_{E}$ is an exact 2-form which is $O(N^{-1})$, and thus vanishes upon integration over the moduli space $\int \Omega_{E} = 0$. . The result of L\'evay \eq{BerryLevay} obtains for $\nu = 1$, $ \mu_{H} = 1/2$, $c_{H} = -2$, and for curvature $R_{0} = -2$. Furthermore, if a variation is taken over the metric while leaving the electromagnetic gauge potential fixed, the result for constant curvature Riemann surfaces is

\begin{equation}\label{KWBerry}
\Omega  = -  \left( \mu_{H}N_{\phi} - \frac{c_{H}}{6}\chi(\Sigma)\right)V \omega_{WP}
\end{equation}
This result follows by relaxing the requirement of constant magnetic field, and treating the electromagnetic gauge potential and spin connection as independent fields.  The subleading correction proportional to the central charge vanishes on the torus. On a sphere, the moduli space is a point and there is no interesting structure arising from adiabatic transport. 

In Ref. \cite{Bradlyn2015}, the Berry curvature on the space of Beltrami differentials was computed for conformal block FQH states (which includes the Laughlin state). They similarly found that the central charge appears subleading to the odd viscosity. 

In this paper, we consider the Berry curvature on the space of punctured spheres. Our main result is that the Berry curvature on this space is controlled by the central charge, which is the leading order contribution. 

\subsection{Motivation}
Motivated by the desire to pull geometric response out of the shadows, we previously \cite{Laskin2016} showed that the subleading effects of geometric response become dominant in the presence of curvature singularities. This fact should be intuitively obvious; at a curvature singularity, $r_{c} = 0$ locally, and the asymptotic expansion in the small parameter $l_{B}/r_{c}$ must necessarily be reorganized. In \cite{Laskin2015}, we argued further that curvature singularities with $R = 4\pi \alpha \delta^{(2)}(z)$ behave as coherent states which possess a localized charge, spin, and exchange statistics, given by 
\begin{equation}\label{cone_anyon}
{\rm charge}\,=\mu_{H} \alpha, \quad {\rm spin} = \frac{c_{H}}{24} \frac{\alpha (2-\alpha )}{1-\alpha}, \quad {\rm statistics} =  \frac{c_{H}}{24} \alpha_{i}\alpha_{j}
\end{equation}
Note that, unlike \eq{KWBerry},\eq{BerryLevayFQH}, the gravitational anomaly is not subleading here, but is indeed the entire effect. 

The quantum Hall effect was studied in a very similar context in \cite{Gromov2016}, where the curvature singularities served to parameterize the moduli on higher genus surfaces, thus giving rise to so-called ``genons" with non-Abelian braiding statistics. For these particular geometries, the spin-statistics theorem holds for genons, and is similarly controlled by the central charge. 

However, these works dealt exclusively with conical singularities. A significant omission was the behavior near a parabolic, or cusp singularity, which are in fact a more natural object mathematically. These are in a sense the most severe singularity you can make, since a cusp corresponds to $\alpha = 1$, which is a cone with zero opening angle. Cusp singularities do have a history in the QH literature, but have not shown up since early work in the 90's. In \cite{Pnueli1994b, Avron1992}, 
adiabatic transport was studied on punctured Riemann surfaces. These surfaces were called ``leaky" since they have finite volume, but include puncture points located at infinity. For this reason, cusps were argued to be idealized leads, thus providing a setting to study quantum scattering. In these works, the punctures were static and not allowed to move. In \cite{Levay1999b}, motivated by questions to do with hard chaos, L\'evay computed Berry curvature for a $\nu = 1$ IQH state on the moduli space of punctured Riemann surfaces. His setting is slightly different from ours, and we discuss the relation in Sec. (\ref{LargeSpin}). 

In this paper, we study the Laughlin state on a punctured Riemann sphere equipped with a complete metric with cusp singularities. We are interested in both the local properties of the state near singularities, as well as the behavior of the QH state under adiabatic motion of the singularities. We find that understanding the local properties is necessary to deduce the adiabatic transport. Our main tool for obtaining analytical results for the Laughlin state is a Ward identity. We also verify some of our results with the exactly solvable integer QH state of free fermions on a punctured disk. 

Although the study of cusp singularities is not as directly relevant experimentally as conical singularities \cite{Schine2015}, we have found that this geometry is both convenient to work with and elucidates the rich mathematical structure of quantum Hall states.


\subsection{Statement of the problem}


We place our QH state on a punctured Riemann sphere $\Sigma = \hat{\mathbb{C}}/X$, where $\hat{\mathbb{C}} = \mathbb{C} \cup \{\infty\}$ is the Riemann sphere and the divisor $X = \{p_{1},...,p_{n-3},0,1,\infty\}$ is a set of points removed from $\hat{\mathbb{C}}$. We consider a finite volume metric $ds^{2} = e^{\phi} dz d\bar{z}$  on $\Sigma$ with constant negative curvature $R_{0}$. 
In the neighborhood of a puncture, the metric looks like the Poincar\'e metric on a punctured disk, namely
\begin{eqnarray}
e^{\phi} &\sim \frac{1}{|z-p_{i}|^{2} \log^{2}|z-p_{i}|^{2}}, \quad z \to p_i, \, \, i = 1, .., n-1\\
& \sim \frac{1}{|z|^{2} \log^{2}|z|^{2}}, \quad \quad \quad \quad z \to p_n = \infty	
\end{eqnarray}
This metric has cusp or parabolic singularities on $X$, which appear as delta-function curvature singularities according to the Liouville equation for the scalar curvature 
\begin{equation}
R(z, \bar{z}) = -4 \partial_{z}\partial_{\bar{z}} \phi  =  R_{0} e^{ \phi} +4\pi  \sum_{i = 1}^{n} \delta^{(2)}(z - p_i)
\end{equation}
For constant $R_{0}$, the Gauss-Bonnet theorem implies the volume $V = - (4\pi/R_{0})(n-2)$. This metric is geodesically complete, and successfully deals with the punctures by sending them to infinity.

 The uniformization theorem guarantees the existence of a conformal map $$w(z): \Sigma \to \mathbb{H}/\Gamma$$ from the punctured sphere to a quotient of the upper half plane $\mathbb{H}$ by a discrete subgroup $\Gamma \subset PSL(2, \mathbb{R})$ of automorphisms of $\mathbb{H}$.

  The metric on $\Sigma$ is the pullback of the Poincar\'e metric $ds^{2} = y^{-2} dx dy$ on $\mathbb{H}$. The map $w(z)$ will take $\Sigma$ to a finite strip on $\mathbb{H}$, sending one puncture to infinity and the rest to the real line. The punctures will appear as cusps in the fundamental domain $\mathbb{H}/\Gamma$, and are the fixed points under the action of $\Gamma$. 

In this paper, we refer to the points $X \in \hat{\mathbb{C}}$ as punctures, and the singularities in the metric as cusps, so that for instance we would describe the metric as having cusp (or parabolic) singularities at the punctures or marked points $X$ (see Sec. \ref{Geometry}).

In this setting, we study the {\it generating functional} for the spin-$s$ Laughlin wave function \cite{Laughlin1983} on a hyperbolic sphere with fluxes $a_{i}$ piercing the punctures at $p_i$,

\begin{equation}\label{Z}
\mathcal{Z}[Q, \phi] = \int_{\Sigma}  \prod_{i<j} | z_{i} - z_{j}|^{2\beta} \prod_{i = 1}^{N} e^{ Q(z_{i}, \bar{z}_{i}) + Q_{a}(z_i, \bar{z}_i)- s \phi(z_i, \bar{z}_i)}  dV_{i},
\end{equation}
where $dV_{i} = (i/2) e^{\phi(z_i, \bar{z}_i)} dz_i \wedge d\bar{z}_{i}$ is the volume form. 
The potential $Q$ is a real-valued function whose behavior at infinity ensures convergence of the integral. In the QH setting, the potential satisfies $-\Delta Q = 2B$, where $B$ is the external magnetic field which we assume to be constant, and the Laplacian is $\Delta = 4e^{-\phi}\partial_{z} \partial_{\bar{z}} $. The potential $Q_{a}(z, \bar{z}) = \sum_{i = 1}^{n} a_{i} \log |z - p_{i}|^{2}$
represents magnetic fluxes $a_{i}$ (in units of the flux quantum $h/e = 2\pi$) threading the punctures at $p_{i}$, and can be understood as sources of singular magnetic field since $B_{a} = -(1/2)\Delta Q_{a} = -2\pi \sum_{i} a_{i} \delta(z - p_{i})$. For integer $a_{i}$ they can also be interpretted as Laughlin quasi-holes. For odd(even)-integer $\beta$, the generating functional $\mathcal{Z} = \int | \Psi_{L} |^{2} $ is the normalization integral for the fermionic (bosonic) Laughlin wave function $\Psi_{L}$. For $\beta = 1$, the wave function is a Slater determinant and the generating functional becomes the determinant of $L^{2}$ norms of single-particle wave functions, computable exactly for certain $Q$.

Finally, the deformation parameter $s$ is referred to often as spin \cite{Can2015, Ferrari2014} (also called ``$j$'' in \cite{Laskin2016,Klevtsov2015}). It is an allowed deformation of holomorphic states on a curved surface, essentially having no effect on the flat plane. However, as we see below, it affects adiabatic transport, and obviously will have a strong influence on the physics near curvature singularities. In the FQH literature, $s = 0$ is often implicitly assumed. Electrons in graphene are described by taking $s = 1/2$. In this paper, we assume real-valued $0<s \le 1$. 

The generating functional is a valuable tool in studying the QH effect for two important reasons. First, variations with respect to $Q$ generate multi-point density correlations, and allow us to study the local structure of the states. Second, $\log \mathcal{Z}$ serves as a K\"ahler potential for the Berry curvature two-form  describing the evolution of the state under adiabatic transport of the punctures $p_{i}$. These two perspectives come together when studying curvature singularities, such as cusps, to reveal the local structure they are imbued with as well as their braiding properties.

There is an important normalizability condition for FQH ground state wave functions. Fixing the total magnetic flux and the surface topology, the total number of particles must satisfy

\begin{equation}\label{Nmax}
N = N_{max} = \nu N_{\phi} + \mu_{H} (2 - 2g) - \sum_{i=1}^{n} \nu a_{i}
\end{equation}
where $N_{\phi} = \frac{1}{2\pi} \int B dV$ is the total magnetic flux and $g$ is the genus of the surface \cite{Wen1992b}. This is the integrated form of the local density in \eq{rhobar}.  For $N < N_{max}$, the mean density forms a droplet with a boundary (or many boundaries). If $N = N_{max}$, the droplet covers the entire surface and is said to be ``incompressible", i.e. lacking a boundary (see Sec. (\ref{holopsi})). The presence or absence of a droplet boundary has very important implications for the large $N$ limit of $\mathcal{Z}$ as well as all correlation functions - we avoid such complications and deal with incompressible droplets.  



Keeping the normalizability condition \eq{Nmax} in mind, we will study the large $N_{\phi}$ (equivalently large $N$) limit of the generating functional, as initiated in \cite{Zabrodin2006} and further elaborated in \cite{Can2014,Can2015, Ferrari2014} (see also the very interesting discussion of the generating functional in the IQH setting \cite{Klevtsov2013}). In particular, we consider the limit $N,N_{\phi} \to \infty$, keeping the volume fixed and the difference $N - \nu N_{\phi} = O(1)$ according to \eq{Nmax}. It turns out to be most natural to develop the asymptotic expansion in the large dimensionless parameter given by the ratio of magnetic field to scalar curvature

\begin{equation}\label{defk1}k = -\frac{4 B}{R_{0}},\end{equation}
(see \eq{defk}).  

As originally noted by Laughlin \cite{Laughlin1983}, the generating functional is equivalent to the partition function for a two-dimensional one-component (Coulomb) plasma (2DOCP)\cite{Forrester2010b}, also known as the Dyson gas \cite{Zabrodin2006}. Under this interpretation, $Q$ is the Coulomb potential created by a background charge distribution which neutralizes the charged particles at $z_{i}$, and $Q_{a}$ is the potential energy of fixed impurity charges at $p_i$. Furthermore, the filling fraction is equal to the temperature of the ensemble (in dimensionless units). The existence of a large $N$ limit is intimately tied to the fact that for high temperatures, the plasma is in a liquid screening phase. The prominent Laughlin filling fractions apparently fall in this range of temperature. Throughout we take for granted that the large $N$ limit exists. For other quantum Hall wave function which do not enjoy such a direct mapping to a 2D plasma, a generalized screening argument has been proposed \cite{Dubail2012}.

\subsection{Main Results}
The generating functional $\mathcal{Z}[p_1,...,p_{n-3}]$ is a real-valued function on the moduli space $\mathcal{M}_{0,n}$ of $n$-punctured genus-0 surfaces of complex dimension ${\rm dim}_{\mathbb{C}} \mathcal{M}_{0,n} = n - 3$. A configuration $(p_{1}, ..., p_{n-3},0,1,\infty)$, modulo the symmetric group on $n$ elements ${\rm Symm}(n)$, describes a point in $\mathcal{M}_{0,n}$. In Sec. (\ref{BerryCurvatureProof}), we prove the following:

\begin{thm}[Berry Curvature]\label{Result1} The generating functional for the spin $s> 0$ Laughlin state is a K\"ahler potential for the Berry curvature two-form on $\mathcal{M}_{0,n}$, which for constant magnetic field and zero threaded flux $a_{i}$ has the large $k$ expansion
\begin{equation}\label{Bcurvature}
 \Omega = i \partial_{\bar{p}} \partial_{p}  \log \mathcal{Z}[\{p_{i}\}] = -\frac{c_{H}}{12\pi } \left( \omega_{WP} - \frac{4\pi^{2}}{3} \omega_{TZ}\right)  + O(k^{-1}),
\end{equation}
%
where $c_{H} = 1 - 12 \mu_{H}^{2}/\nu$ is the central charge \eq{rho_grads}. Here, $\partial_{p} = \sum_{i = 1}^{n-3}dp_{i} \partial_{p_{i}}$ is the Dolbealt operator, $\omega_{WP}$ is the Weil-Petersson (W-P) metric, and $\omega_{TZ}$ is the Takhtajan-Zograf (T-Z) metric.
\end{thm}
We recall the definitions of the W-P and T-Z metrics in Sec. (\ref{modulispace}). A remarkable property of this formula is that there are no contributions of order $k^{2}$ or $k$, as one might naively expect from the large $k$ expansion of the generating functional in \cite{Zabrodin2006, Can2015}, as well as the corresponding result on smooth higher genus surfaces in \cite{Levay1997a}. We discuss the discrepancy in more detail in Sec. (\ref{LargeSpin}). The formula for the generating functional $\mathcal{Z}$ is given in Eq. \eq{gen_fun_full}. The special combination of W-P and T-Z metrics appearing in \eq{Bcurvature} is related to Quillen's determinant of the Laplacian on one-forms \cite{Takhtajan1991}. The connection to Quillen's metric was also pointed out in \cite{Levay1997a,Levay1999b,Klevtsov2015a}.

In Sec.(\ref{BerryCurvatureProof}), we show that this result is modified for non-zero fluxes $a_{i}$ by the addition of the curvature form
\begin{equation}
	\Omega_{a} = - \sum_{i<j} 4 \nu a_{i}a_{j} \left( \delta(p_i - p_j) \frac{i}{2} dp_{i} \wedge d\bar{p}_{j}\right)  - \frac{8\pi}{3} \sum_{i} h_{a_{i}} \omega_{TZ, i},
\end{equation}
where $\omega_{TZ,i}$ is the T-Z metric of cusp $p_i$, and $h_{a_i} = a_i (2\mu_{H} - \nu a_{i})/2$ is the conformal dimension of the Laughlin quasi-hole \cite{Can2015}.

For $n = 4$, the total integrated Berry curvature, or Chern number, becomes a combination of the volumes of the W-P \eq{volumeWP} and T-Z metrics \eq{volumeTZ}, and we find

\begin{equation}\label{cohom}
	\frac{1}{2\pi} \int \Omega  = \frac{c_{H} }{8} .
\end{equation}
Thus the central charge is a topological invariant \cite{Klevtsov2015,Bradlyn2015} and the Berry curvature belongs to a rational cohomology class \cite{Klevtsov2015a,Wolpert2005}.

The first variation of the generating functional with respect to the potential $Q$ gives the mean particle density

\begin{equation}
	\langle \rho(z, \bar{z}) \rangle = e^{ - \phi} \frac{\delta \log \mathcal{Z}}{\delta Q(z, \bar{z})},
\end{equation}
where $\rho(z) = \sum_{i = 1}^{N} \delta(z - z_{i}) $ is the microscopic particle density. In the large $N$ limit, the density approaches the mean value away from the punctures

\begin{equation}
	\langle \rho(z, \bar{z}) \rangle = \frac{\nu}{2\pi l_{B}^{2}} + \frac{\mu_{H}}{4\pi} R + O(l_{B}^{2}), \quad |z - p_{i} | >> l_{B} = B^{-1/2}.
\end{equation}
Near the punctures, the large $N$ expansion becomes highly singular. This occurs even though the density is perfectly smooth for all finite $N$. If we insist on keeping the large $N$ limit, we must then treat the density as a distribution. To this end, we look at {\it moments} and {\it sum rules} of the density around a singularity. 

\begin{defn}[Moments and Sum Rules]
Let $D_{\epsilon}(p_{i})$ be an open disk of radius $\epsilon$ centered on the puncture $p_{i}$. The moments of the mean density around a puncture are defined

\begin{equation}\label{moments}
	M_{n}(p_i) = \lim_{\epsilon \to 0} \lim_{l_{B} \to 0}\frac{\Gamma(k+2s-1)}{\Gamma(k+2s - 1-n)} \int_{D_{\epsilon}(p_i)}  \frac{(-1)^{n}}{\log^{n}|z-p_i|^{2}} \left( \langle \rho \rangle - \bar{\rho}\right) dV_{z}.
\end{equation}
The sum rules involve the function $W = Q + (1 - s) \phi$, and are defined 

\begin{equation}\label{sumrules}
	\mathcal{L}_{n}(p_i) = \lim_{\epsilon \to 0} \lim_{l_{B} \to 0}\int_{D_{\epsilon}(p_i)}  (z - p_i)^{n+1} \partial_{z} W(z) \left( \langle \rho \rangle - \bar{\rho}\right) dV_{z}.
\end{equation}
In particular, the {\bf charge} of the cusp singularity at $p_{i}$ is given by the zeroth moment $M_{0}$, while the {\bf conformal dimension}, or simply dimension, of the singularity is given by the first moment $M_{1}$. 

The definition of the sum rules is motivated by their connection to the Ward identity discussed in Sec.(\ref{WardIdentity}), and bears a resemblance to Virasoro constraints. In particular, we will single out the {\bf translation} $\mathcal{L}_{-1}$ and the {\bf dilatation} $\mathcal{L}_{0}$ sum rules. 
\end{defn}

The order of the limits in these definitions requires comment. Taking magnetic length to zero (equivalently $N \to \infty$) ensures that we remove finite size effects and extract the intrinsic features of the QH state near the singularity. The second limit ensures that we recover only the {\it intensive} piece which does not depend on the integration domain. 

The meaning of the charge is clear - it is simply the total number of electrons which accumulates at the cusp. The conformal dimension requires further motivation which we provide below and in later sections. For now, we briefly comment that it involves the expectation of the generator of dilatations $\ell_{0} = \sum_{i} z_{i} \partial_{z_{i}}$ in coordinates for which the cusp is at the origin. The subtraction by $\bar{\rho}$ and the $\epsilon \to 0$ limit essentially extract the scaling behavior of the singularity as in \cite{Laskin2016}.

We now state the main results, proved in Sec. (\ref{WardIdentity}):

\begin{prop}\label{momentsprop}
For the the Laughlin state, the charge and dimension of a cusp singularity threaded by a flux $a$ are given by  
\begin{eqnarray}\label{momentsFQH}
 {\rm Charge:} \quad &M_{0} = \mu_{H}   - \nu a,\label{charge} \\
  {\rm Dimension:}\quad  &M_{1} = \frac{c_{H}}{24}+ \frac{a}{2} (2 \mu_{H} - \nu a).\label{dimension}
\end{eqnarray}

	
\end{prop}
Despite the fact that the wave function vanishes at $p_i$, according to \eq{momentsFQH}, for $0<s < 1/2\nu$ there can be {\it positive} charge at the cusp. Clearly, this is only possible if the charge is somehow distributed in a region surrounding the puncture. The details of this distribution, while very interesting, will not bother us in this paper, since we focus on integrated quantities. We are content to say that this accumulated charge belongs to the cusp if it is concentrated in a region whose area shrinks to zero as the $l_{B} \to 0$. 

The conformal dimension $M_{1}$ is the cusp analog of the angular momentum of conical singularities defined in \cite{Laskin2016}. The first clear connection appears in the flux part of \eq{dimension} $h_{a}$, which was argued in \cite{Can2015} to be the conformal dimension of a quasi-hole (it is also called spin; see \cite{Gromov2016} for a careful review of the literature on this delicate issue). Since the central charge contribution $c_{H}/24$ enters on the same footing, we are tempted to call this the conformal dimension of the cusp singularity. It differs slightly from the angular momentum defined in \cite{Laskin2016}, which if applied naively to cusp singularities appears divergent (e.g. take \eq{cone_anyon} for $\alpha = 1$) \footnote{This difference can be traced to the fact that in the present paper, we connect dilatations to sum rules involving $z \partial_{z} W$, whereas in \cite{Laskin2016}, the sum rule involves only $W$. Using $W$ instead of $z \partial_{z} W$ makes sense only in the flat cone setting where $z \partial_{z} W \propto W$, and would be the incorrect object to look at for more general surfaces.}. The dimension as we define it here, which parameterizes the effect of dilatations (Prop. \eq{sumrulesprop}) and is thus related to the scaling dimension of a cusp, is a more natural object for the present setting and is the cusp analog of angular momentum or spin in \eq{cone_anyon}.

In Sec. (\ref{sum_from_res2}), we prove the following sum rules used in the derivation of the Berry curvature, 

\begin{prop}\label{sumrulesprop}

The local density around a puncture $p_{i}$ with threaded flux $a_{i}$ satisfies the sum rules for dilatation and translation, given to leading non-vanishing order in large $k$ by
\begin{eqnarray}
{\rm Dilatation: } \quad &\mathcal{L}_{0} = - M_{1} + (s-1) M_{0} 	\\
{\rm Translation:} \quad &\mathcal{L}_{-1} = -2 M_{1} \gamma_{i} 
\end{eqnarray}
where $M_{0}$ is the charge \eq{charge}, $M_{1}$ is the dimension \eq{dimension}, and  $\gamma_{i}$ are the {\rm accessory parameters} of the Fuchsian uniformization of $\Sigma$, appearing in the Schwarzian derivative of the developing map $w(z)$ \eq{stress}. 
	
\end{prop}

The logical order of these propositions is actually reversed. We first compute the sum rules in Sec. (\ref{sum_from_res2}), and from these deduce the moments using the formulas in Prop. (\ref{sumrulesprop}) derived in Sec. (\ref{mom_from_sum})

We refer to the next section (\ref{Geometry}) for the complete definition of the accessory parameters, as well as other mathematical background. Section (\ref{physics}) contains the physics background, with a discussion of holomorphic wave functions in the lowest Landau level. Section (\ref{WardIdentity}) is devoted to developing the Ward identity that is used to prove Props. \eq{momentsprop} and \eq{sumrulesprop}. The main theorem \eq{Result1} is proved in Sec.(\ref{BerryCurvatureProof}), where we apply the sum rules to find a variational formula for the generating functional. This section contains a more detailed discussion of the physical meaning of our results. Finally, we check our results for the Laughlin state against exact results for the completely filled LLL on a punctured disk in Sec. (\ref{IQHexact}), where we compute all moments of the density in closed form.

\section{Geometry of singular Riemann surfaces}\label{Geometry}

In this section we provide the mathematical background necessary to obtain our main results. As such, this section does not contain any new material, and we have preferred a terse, pragmatic collection of useful facts with appropriate references to more pedagogical treatments of the theory of Riemann surfaces.

\subsection{The Punctured Sphere}

Let $\hat{\mathbb{C}} = \mathbb{C} \cup \{ \infty\}$ be the Riemann sphere, and let $X = \{p_{1}, ..., p_{n}\}$ be a set of points in $\hat{\mathbb{C}}$, called a divisor. The n-punctured Riemann sphere is denoted $\Sigma = \hat\mathbb{C}/ X$, consists of the surface with the points in $X$ removed. An element in the automorphism group $SL(2, \mathbb{C})$ is a M\"obius transformation that fixes three points.  A standard normalization is to choose the last three points as $p_{n-2} = 0$, $p_{n-1} = 1$ and $p_{n} = \infty$. It is for this reason that the configuration space of punctured sphere has complex dimension $n -3$.  

A healthy metric on a punctured Riemann surface will send the points in $X$ infinitely far from any other point in $\Sigma$, thus leaving the geometry geodesically complete but noncompact. This is achieved quite nicely using hyperbolic metrics, which have constant negative curvature $R_{0}$. A theorem of Poincar\'e states that there exists a unique hyperbolic metric on a punctured Riemann surface whose conformal factor satisfies the Liouville equation away from the punctures

\begin{equation}
	4\partial_{z} \partial_{\bar{z}} \phi =  -R_{0} e^{ \phi},
\end{equation}
and which has a {\it parabolic} or {\it cusp singularity} at the punctures, where the metric behaves asymptotically like
\begin{eqnarray}\label{cusp}
	\phi &= - \log |\zeta_i|^{2} - 2 \log \left | \log | \zeta_i / a_{i}(1)|^{2} \right| + o(1), \quad z \to p_{i}, \quad \zeta_{i} \equiv z - p_{i}. \\
	\phi &  =- \log |z|^{2} - 2 \log \log |z/a_{n}(-1)|^{2}, \quad z \to p_n = \infty
\end{eqnarray}
The leading behavior implies the geodesic distance $d(z,  p_{i}) = \infty$ for all $z \in \Sigma$. The complex numbers $a_{i}(1)$ and $a_{n}(-1)$ are Fourier coefficients defined below in Eq. (\ref{KleinH}). They are functions of the coordinates $p_{i}$ and set a local scale near each puncture. On $\hat\mathbb{C}$, the metric has a curvature singularity at the puncture, with a scalar curvature locally given by

\begin{equation}
	R(z, \bar{z}) = - 4e^{-\phi} \partial_{z} \partial_{\bar{z}} \phi = R_{0} +  4\pi \delta^{(2)}(z - p_{i})
\end{equation}
Integrating the LHS and applying the Gauss-Bonnet theorem for the sphere $\hat{\mathbb{C}}$ implies the volume of $\Sigma$

\begin{equation}\label{Volume}
V = - \frac{4\pi}{R_{0}} ( n - 2)
\end{equation}
Clearly, for genus $g$, this generalizes to $V = - (4\pi/R_{0})(n + 2g - 2)$.

The uniformization theorem \cite{Donaldson2011, FarkasKra} states that the upper half plane $\mathbb{H}$ is the universal (i.e. simply-connected) cover of the punctured hyperbolic Riemann sphere $\Sigma$. The {\it developing map} \cite{Yoshida1987, Kuga1993}
\begin{equation}\label{developing_map}
	w(z): \Sigma \to \mathbb{H}/\Gamma
\end{equation}
 is a multi-valued conformal map between $\Sigma$ and a fundamental domain on the upper half plane, denoted by the quotient space $\mathbb{H}/\Gamma$. Here, $\Gamma$ is a discrete subgroup of the full automorphism group $SL(2, \mathbb{R})$ of the covering space $\mathbb{H}$. These are also known as Fuchsian groups. On a hyperbolic sphere, the Fuchsian group $\Gamma$ consists of a collection of $n$ {\it parabolic} M\"obius transformations $\Gamma = \{ P_{i}, \,\, i  = 1, ... , n\}$. A M\"obius transfromation 
 \begin{equation}\label{mobiusP}
 	P(w) = \frac{a w + b}{c w + d}, \quad ad - bc = 1
 \end{equation}
 is said to be parabolic if $|a +d| = 2$. The fixed points of $P_{i}(z)$ for $i \ne n$ lie on the real axis and correspond to cusp or parabolic singularities. The map $P_{n}(w) = w+1$ has $w_n =  \infty$ as a fixed point \cite{FarkasKra}. For the punctured sphere, these maps can be explicitly parameterized \cite{Zograf1988h}. 
  
The uniformization theorem enables us to write a metric on $\Sigma$ exclusively in terms of the developing map $w(z)$. This is accomplished by taking the simple Poincar\'e metric on the universal cover
 
 \begin{equation}
ds^{2} = \frac{\lambda }{4({\rm Im}\, w )^{2}} |dw|^{2}, \quad \lambda \in \mathbb{R}, \quad w \in \mathbb{H}/\Gamma
 \end{equation}
which has constant negative curvature $R = -8/\lambda$. The pullback under $w(z)$ will induce a constant negative curvature metric on the punctured sphere

\begin{equation}\label{uniform_metric}
	ds^{2} \equiv e^{\phi(\Sigma)} |dz|^{2} = \frac{\lambda |w'(z)|^{2}}{4\left( {\rm Im}\,  w(z) \right)^{2}} |dz|^{2}, \quad z \in \Sigma.
\end{equation}
Consequently, the asymptotes \eq{cusp} must follow from the behavior of $w(z)$ at the punctures. To determine this behavior, define the {\it Schwarzian of the metric} as

\begin{equation}\label{schwarz1}
\mathcal{S}[\phi] \equiv \phi_{zz} - \frac{1}{2} ( \phi_{z})^{2}	 = q(z).
\end{equation}
For a metric with cusp singularities \eq{cusp}, a classic result of Schwarz states that
\begin{equation}\label{stress}
	q(z) = \sum_{i = 1}^{n-1} \left(\frac{1/2}{(z - p_{i})^{2}} + \frac{\gamma_{i}}{(z - p_{i})} \right), \quad q(z) = \frac{1}{2 z^{2}} + \frac{\gamma_{n}}{z^{3}}, \quad z \to \infty.
\end{equation}
The residues $\gamma_{i}$ of the first order poles are known as {\it accessory parameters}, and are not fully determined by the asymptotics near the cusps. The asymptotic form at infinity implies the following three constraints on the accessory parameters \cite{Zograf1988h}

\begin{eqnarray}\label{WIschwarz}
	\sum_{i=1}^{n-1} \gamma_{i} = 0, \quad  \sum_{i = 1}^{n-1} \gamma_{i} p_i = \frac{2 - n}{2}\quad \sum_{i = 1}^{n-1} \left( p_i + \gamma_{i} p_{i}^{2}\right) = \gamma_{n}.
\end{eqnarray}
This is reminiscent of the conformal Ward identity in conformal field theory, with $q(z)$ playing the role of the stress-energy tensor. Eq. \eq{WIschwarz} fixes three of the accessory parameters, say $\gamma_{n-2}$, $\gamma_{n-1}$ and $\gamma_{n}$,  in terms of the $n-3$ other accessory parameters. For $n=3$, Eq. \eq{WIschwarz} fixes the accessory parameters completely. Fixing only $p_{3} = \infty$, for $n = 3$ we find $\gamma_{1} = - \gamma_{2}= -1/2(p_{1} - p_{2})$, and $\gamma_{3} = (p_{1} + p_{2})/2$. For $n>3$, the leading asymptotic behavior of the accessory parameters as punctures merge (i.e. near the boundary of moduli space) is known \cite{Zograf1990a}
\begin{eqnarray} \label{accessory_asymp}
	\gamma_{i}(p_1, ..., p_{n-3})& \nonumber \\
	& = - \frac{1}{2(p_{i} - p_{j})}, \quad p_i \to p_{j}\\
	 &= - \frac{1}{2 p_{i}}, \quad p_{i} \to p_{n} = \infty\\
	&= \gamma_{i}(p_{1}, ..., p_{j-1}, p_{j+1}, ..., p_{n-3}), \quad p_{j} \to p_{k}, \quad k, j \ne i
\end{eqnarray}
Plugging (\ref{uniform_metric}) into \eq{schwarz1} implies that

$$\mathcal{S}[\phi(z)] = \frac{w'''(z)}{w'(z)} - \frac{3}{2} \left( \frac{w''(z)}{w'(z)}\right)^{2} \equiv \{ w, z\},$$
which is the Schwarzian derivative of $w(z)$. The developing map must then satisfy the nonlinear differential equation $\{ w(z), z\} = q(z)$. The solution is found by writing $w(z) = f_{1}(z)/f_{2}(z)$, where $f_{i}$ are linearly independent solutions of the Fuchsian ODE $f'' + \frac{1}{2}q(z) f = 0$ with monodromy group $\Gamma$ (see e.g. \cite{Yoshida1987, Kuga1993}). Finding the conformal map $w(z)$ by employing its connection to this Fuchsian ODE is known as Fuchsian uniformization. For $n = 3$, the accessory parameters are fully determined, and the map is given in terms of hypergeometric functions. An example of the fundamental domain for the thrice-punctured sphere is shown in Fig. (\ref{fundomain}).

 \begin{figure}[htbp]
 \begin{center}
   \includegraphics[scale=0.75]{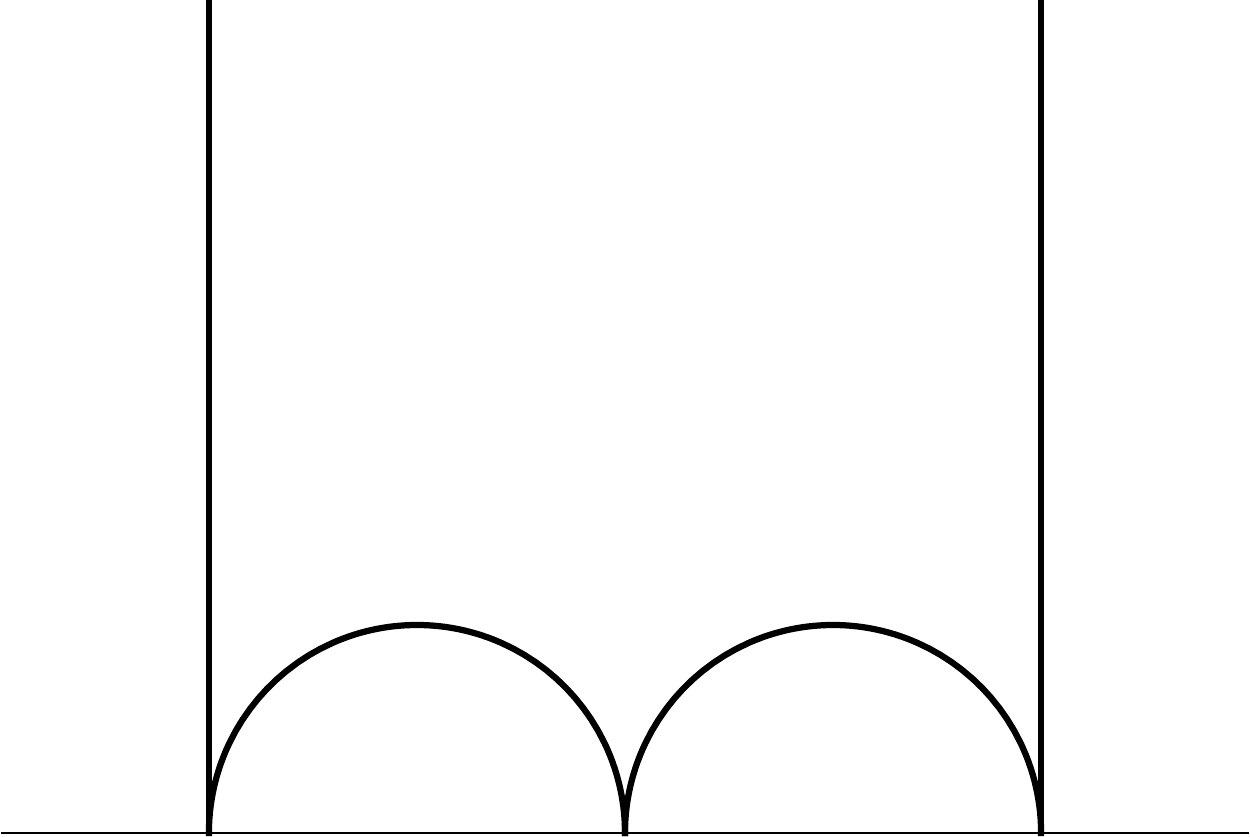}
 %
   \caption{ The thrice punctured sphere is equivalent to the strip on the upper-half plane, with cusps that touch the real axis. }
      \label{fundomain}
      \end{center}
   \end{figure}  

The inverse of developing map, known as the Klein Hauptmodul \cite{Zograf1988h},
\begin{equation}\label{Klein_map}
	J(w): \mathbb{H}/\Gamma \to \Sigma,
\end{equation}
is a single-valued function on the universal cover invariant under $\Gamma$, i.e. $J(\gamma w ) = J(w)$, where $\gamma \in \Gamma$ acts by $\gamma w = (a w + b)/(c w + d)$. The map $J(w)$ enjoys a Fourier expansion near the singularities \cite{Park2015}. Let us assume that $J^{-1}(z_{i}) = \infty$ takes the cusp $z_{i}$ to infinity. Then the Fourier expansion reads \cite{Zograf1988h}
 
 \begin{eqnarray}\label{KleinH}
 J(w) &= z_{i} + \sum_{k = 1}^{\infty} a_{i}(k) e^{ 2\pi i k w} , \quad {\rm for}\,\, P_{i},\\
 J(w) &=  \sum_{k = -1}^{\infty} a_{n}(k) e^{ -2\pi i k w} , \quad {\rm for} \,\, P_{n}.
 \end{eqnarray}
The first coefficients $a_{i}(1)$ and $a_{n}(-1)$ appear in the leading asymptotics of the cusp metric (\ref{cusp}). The cusp points $x_{i} \in \mathbb{H}$ are sent to $z_{i} \in \Sigma$, and lie on the real line. 
From the expansion \eq{KleinH}, the developing map $w = J^{-1}$ can be seen to have the following expansion near a cusp as $z \to p_i$ \cite{Zograf1988h}
\begin{eqnarray}\label{dev_map}
	 w(z) = \frac{1}{2\pi i} \left[ \log (\zeta_{i}/a_{i}(1)) - \frac{ a_{i}(2)}{(a_{i}(1))^{2}} \zeta_{i} + O(\zeta_{i}^{2}) \right] , \quad i \ne n,\\
	  w(z) =  - \frac{1}{2\pi i} \left[ \log (z/ a_{n}(-1)) - \frac{a_{n}(0)}{z}  + O(z^{-2})\right], i = n. \nonumber
%
\end{eqnarray}
Taking the Schwarzian derivative and comparing to \eq{stress}, one finds that the accessory parameters can be expressed in terms of the Fourier coefficients

\begin{equation}
	\gamma_{i} = -\frac{a_{i}(2)}{(a_{i}(1))^{2}} ,\, \, i \ne n, \quad \gamma_{n} = a_{n}(0).
\end{equation}

Another useful object will be the local K\"ahler potential $K$ for the metric, defined (up to a constant) by the relation $\partial_{z} \partial_{\bar{z}} K  = e^{\phi}$. For the metric (\ref{uniform_metric}), we may use

\begin{equation}\label{uniform_kahler}
	K(z, \bar{z}) = -  \lambda \log {\rm Im}\left[ w(z)\right].\end{equation}
Clearly, we have $\phi =  2K/\lambda  + \log |w'(z)|^{2}$. Using the expansion for \eq{dev_map},  we find that near a cusp, as $z \to p_i$, 
\begin{eqnarray}\label{kahler}
K  = - \lambda \log \Big| \log \left| \zeta_{i} / a_{i}(1)\right| \Big| + \frac{2 \lambda {\rm Re}(\gamma_{i} \zeta_i )}{\log |\zeta_{i}/a_{i}(1)|^{2}} , \quad {\rm for}\, i \ne n,\,    \label{kahlerpi}\\
K =  \frac{ V}{\pi}\log |z|^{2} -  \lambda\log \Big| \log |z/a_{n}(-1)| \Big| + \frac{2\lambda {\rm Re}(\gamma_{n}/z)}{\log|z/a_{n}(-1)|^{2}} , \, \, \, i = n  ,\label{kahlerinf}
\end{eqnarray}
we keep the subleading correction (the last term in \eq{kahlerpi} and \eq{kahlerinf}) for future reference, but point out that it vanishes in the limit $\zeta_i = 0$ (or $z = \infty$). The asymptote at infinity can also be written $K \to   \frac{\lambda}{2}( n-2)\log|z|^{2}$. It follows by requiring $\int_{\mathbb{C}} \partial_{z} \partial_{\bar{z}} K d^{2} z =  V$.

\subsubsection{Variational formulas} We will need to make extensive use of the following variations of the K\"ahler potential and the metric with respect to the moduli $p_i$ of the metric. Most of these formulas are taken either verbatim or with minor modifications from the reference \cite{Zograf1988h}. 
In the neighborhood of a puncture, we will need
\begin{eqnarray}\label{cuspexplicit}
z \to p_j , \nonumber\\
\partial_{p_i} \phi  = - \delta_{ij} \partial_{z} \phi  +   \frac{2}{\log |\zeta_{j}/a_{j}(1)|^{2}}\partial_{p_{i}} \log |a_{j}(1)|^{2} ,\\
\partial_{p_{i}} K = -\, \delta_{ij} \partial_{z} K +   \frac{\lambda}{\log |z/a_{j}(1)|^{2}} \partial_{p_{i}} \log |a_{j}(1)|^{2} ,\label{kahlermod}
\end{eqnarray}
and at infinity, 
\begin{eqnarray}\label{cuspinf}
z \to p_n = \infty,  \nonumber\\
\partial_{p_i} \phi  = \frac{2}{\log|z/a_{n}(-1)|^{2}}  \partial_{p_i} \log |a_{n}(-1)|^{2} + O\left(z^{-1}\log^{-2}|z|\right) ,\\
\partial_{p_i} K  =  \frac{\lambda }{\log|z/a_{n}(-1)|^{2}} \partial_{p_i} \log |a_{n}(-1)|^{2} + O(z^{-1} \log^{2}|z|).\label{kahlermodinf}
\end{eqnarray}


These can be seen in a non-rigorous way directly from differentiating the asymptotic formulas \eq{cusp} and \eq{kahler}. Their proof requires an application of the Ahlfors lemma as stated in \cite{Zograf1988h},

\begin{equation}\label{ahlfors1}
\partial_{p_i} \phi+ \dot{F}^{i}  \partial_{z}\phi+ \dot{F}_{z}^{i} = 0. 
\end{equation}
This formula allows one to trade variations with respect to $p_i$ with variations with respect to the complex structure moduli through the quasiconformal map $F$ and its variations $\dot{F}$, defined in \eq{quasiF} in Sec. (\ref{modulispace}). Eqs. \eq{cuspexplicit} and \eq{cuspinf} follow from the asymptotic expansion of the quasiconformal map \cite{Park2015}
\begin{eqnarray}\label{Fasymp}
\dot{F}^{i}(z) &= \delta_{ij} + (z - p_j) F_{z}^{i}(p_j) + O( \zeta_{i}^{2}) , \quad z \to p_j,\\
\dot{F}_{z}^{i}(p_j) & = \partial_{p_i} \log |a_{j}(1)|^{2}, \quad \dot{F}_{z}^{i}(p_n) = \partial_{p_i} \log | a_{n}(-1)|^{2}.
\end{eqnarray}
We find that the local K\"ahler potential satisfies 

\begin{equation}\label{ahlfors2}
\partial_{p_i} K + \dot{F}^{i} \partial_{z} K = 0.
\end{equation} 
We list some more refined asymptotics in the \ref{App_asymp}, since they will be used later on. 

\subsection{Symplectic forms on Moduli space}\label{modulispace}
Using the metric to probe geometric response means we must be allowed to vary it. An illuminating class of deformations involve motion in the moduli space of a Riemann surface. We briefly review this concept for a punctured sphere, again with the purpose of introducing notation, and discuss two important symplectic forms on this space: the Weil-Petersson (W-P) metric and the Takhtajan-Zograf (T-Z) metric. 

There are different constructions of the moduli space of a punctured Riemann sphere. In this paper, we prefer to describe moduli space in terms of the puncture coordinates. However, we need to make use of concepts in Teichm\"uller theory to connect the motion of punctures to the changes of complex structure. This allows a projection of symplectic forms on the Teichm\"uller space to the configurational moduli space. This is the perspective taken in \cite{Zograf1988h,Park2015}, which we review below.

\subsubsection{Moduli Space}

The automorphism group on the punctured Riemann sphere consists of global conformal transformations $SL(2, \mathbb{C})$ combined with the symmetric group ${\rm Symm}(n)$ on $n$ elements acting by interchanging the punctures.

The configuration space $M_{n}$ of punctures consists of divisors $X$ which are inequivalent under $SL(2, \mathbb{C})$. Since a M\"obius transformation can be used to fix the position of any three punctures, an element of $M_{n}$ is determined by the remaining $n-3$ complex coordinates. We use the standard normalization which takes $p_{n-2}= 0$, $p_{n=1} = 1$, and $p_{n} = \infty$, so that the $n$-punctured sphere can be viewed as the $n-1$-punctured complex plane. The configuration space is then  


\begin{equation}
M_{n} = \{ (p_1, .., p_{n-3}) \in \mathbb{C}^{n-3}| p_{i} \ne 0, 1, \quad p_{j} \ne p_{k}\, {\rm for} \, j \ne k\}.
\end{equation}

The moduli space of a Riemann surface with genus-g and $n$ punctures is denoted $\mathcal{M}_{g,n}$. For the punctured sphere, it is constructed by taking the quotient of $M_{n}$ with the symmetric group ${\rm Symm}(n)$, identifying configurations in $M_{n}$ which are related by interchanging punctures $p_{i} \to p_{j}$, for all $i, j$. Thus, the moduli space of an n-punctured Riemann sphere is 

\begin{equation}\label{moduli_p}
	\mathcal{M}_{0,n} = M_{n}/{\rm Symm}(n).
\end{equation}

\subsubsection{Teichm\"uller Space} 

 The moduli space \eq{moduli_p} can be identified with the space of complex structures on $\Sigma$ \cite{Zograf1988h,Zograf1990a,Park2015}. This is accomplished using Teichm\"uller theory \cite{Teichmuller,Farb}. The Teichm\"uller space is the quotient of the space of constant curvature hyperbolic metrics with the group of orientation preserving diffeomorphisms isotopic to the identity, denoted

\begin{equation}
	\mathcal{T}(\Sigma) = {\rm Metrics}(\Sigma)/{\rm Diff}_{0}(\Sigma).
\end{equation}
What this leaves is a finite dimensional complex manifold, locally isomorphic to $\mathbb{C}^{n-3}$, which describes metrics not related by conformal coordinate transformations. The tangent space to $ \mathcal{T}$ consists of the Beltrami differentials  $\mu(z, \bar{z})$, which are $(-1,1)$ tensors $\mu = \mu_{\bar{z}}^{z} d\bar{z} (d z)^{-1}$ satisfying $\mu(\gamma z) \overline{\gamma'(z)} /\gamma'(z) = \mu(z)$ for $\gamma \in \Gamma$. Harmonic Beltrami differentials satisfy $\partial_{z} (e^{\phi} \mu ) = 0$, and thus take the form $\mu = e^{ - \phi} \overline{q(z)}$, where $q(z)$ is a quadratic differential. By the Riemann-Roch theorem, the space of quadratic differentials on $\Sigma$ has dimension $n-3$. Let us denote then the basis of harmonic Beltrami differentials $\mu_{1}, \mu_{2}, .... , \mu_{n-3}$. Using local coordinates $(\tau_{1}, ..., \tau_{n-3}) \in \mathbb{C}^{n-3}$, known as Bers coordinates, an element of Teichm\"uller space is a metric 
\begin{equation}
	ds^{2} = e^{\phi_{\mu}} | dz + \mu d\bar{z}|^{2},
\end{equation}
where $\mu = \sum_{i = 1}^{n-3} \tau^{i} \mu_{i}$, and $\phi_{\mu}$ is an unspecified conformal factor. This describes the holomorphic mapping between metrics in $\mathcal{T}$ and complex Bers coordinates in $\mathbb{C}^{n-3}$ \cite{Nelson1987}.

A quasiconformal map $f^{\mu}$ which satisfies the Beltrami equation
\begin{eqnarray}
	\partial_{\bar{z}} f^{\mu} &= \mu \partial_{z} f^{\mu}, \quad z \in \mathbb{H},
\end{eqnarray}
describes a homeomorphism $f^{\mu}: \mathbb{H} \to \mathbb{H}$ and brings the metric to conformal form $|df|^{2} 
	= |\partial_{z} f|^{2} \left| dz  + \mu d\bar{z}\right|^{2}$. A quasiconformal map changes the complex structure, and thus metrics related by a quasiconformal map represent different points in $\mathcal{T}(\Sigma)$.  
The Klein map $J$ \eq{Klein_map} 
induces a quasiconformal map $F^{\mu}: \Sigma \to \Sigma$ between punctured spheres.



The moduli space is obtained by identifying points in $\mathcal{T}(\Sigma)$ conjugate under the mapping class group of $\Sigma$, denoted ${\rm Mod}(\Sigma)$, and is given by the quotient $\mathcal{T}(\Sigma)/{\rm Mod}(\Sigma)$ \cite{Farb}. This construction of the moduli space is equivalent to that in configurational space \eq{moduli_p} under the covering map

\begin{equation}
\Psi: \mathcal{T} \to \mathcal{M}_{0,n}	.
\end{equation}
Choosing a base point $\tau_{0} \in \mathcal{T}$, we have that the Bers coordinates are sent to the puncture coordinates $\Psi(\tau_{0}^{i}) = p_{i}$. Changing the complex structure moves $\tau_{0}^{i} \to \tau_{\mu}^{i}$ and sends $p_i \to F^{\mu}(p_i)$, giving $\Psi(\tau_{\mu}^{i}) = F^{\mu}(p_i)$. This connects the motion of punctures  $ (p_1, ..., p_{n-3}) \in \mathcal{M}_{0,n}$ with deformations of metrics in $\mathcal{T}(\Sigma)$ via the quasiconformal map $F^{\mu}$. The differential $d\Psi(\mu_{i}) = \partial_{p_i}$ maps the tangent space of $\mathcal{T}$ to that of $\mathcal{M}_{0,n}$ \cite{Zograf1988h}. Finally, along the basis vectors $\mu_{i}$ we denote $F^{\mu_{i}} \equiv F^{i}$, and variations in the direction of $\mu_{i}$ are

\begin{equation}\label{quasiF}
	\dot{F}^{i} = \partial_{\epsilon} F^{\epsilon \mu_{i}}\Big|_{\epsilon = 0}.
\end{equation}
This variation is used in the Alfhors lemma (\ref{ahlfors1}) to connect motion in the configurational moduli space $\mathcal{M}_{0,n}$ with deformations of the complex structure moduli $\tau$. 

\subsubsection{ Weil-Petersson metric} A natural inner product on the space of metrics is given by \cite{DHoker1988}

\begin{equation}
	|| \delta g||^{2} = \int_{\Sigma}  g_{ac}g_{bd} \delta g^{ab} \delta g^{cd} dV.
	\end{equation}
	For deformations of the metric which lie within $\mathcal{T}(\Sigma)$ and thus only change the complex structure, this inner product is projected to the Petersson inner product on Beltrami differentials
	
	\begin{equation}
		\langle \mu_{i}, \mu_{j} \rangle_{WP} = \int_{\Sigma} \mu_{i} \bar{\mu}_{j} \, dV.
	\end{equation}
A symplectic two-form on $\mathcal{T}(\Sigma)$ can be constructed using this inner product 

\begin{equation}
	\omega_{WP} = \sum_{i,j} \langle \mu_{i}, \mu_{j} \rangle_{WP}\, \frac{i}{2} d\tau_{i} \wedge d\bar{\tau}_{j} =  \frac{i}{2} \int_{\Sigma} \mu(z,\bar{z})\wedge \overline{\mu(z,\bar{z})} \, dV_{z}. 
\end{equation} 
This is the Weil-Petersson metric \cite{Teichmuller}. It is a measure of distance on the moduli space, indicating roughly how ``different" two complex structures are.

The W-P metric is well-known to be K\"ahler. When projected onto the moduli space $\mathcal{M}_{0,n}$ \eq{moduli_p}, it can be expressed explicitly as 

\begin{equation}
	\omega_{WP} =  \frac{i}{2} \partial_{\bar{p}} \partial_{p} S_{L} =  -\sum_{i,j = 1}^{n-3}  \frac{\partial^{2}S_{L}}{\partial \bar{p}_{j} \partial p_{i}}  \, \frac{i}{2} dp_{i} \wedge d\bar{p}_{j} ,
\end{equation}
with the K\"ahler potential $S_{L}$ given by the regularized Liouville action

\begin{equation}
	S_{L} =\lim_{\epsilon \to 0} \int_{\Sigma_{\epsilon}} \left(| \partial_{z} \phi |^{2} +  e^{\phi}\right) d^{2} z + 2\pi n ( \log \epsilon + 2 \log | \log \epsilon|),
\end{equation}
where $\Sigma_{\epsilon}$ is the punctured sphere with small $\epsilon$ balls removed around each point in $X$. This fact was proven in a series of beautiful papers \cite{Zograf1988h,Zograf1988c,Takhtajan2002,Zograf1990a}, which also proved the so-called Polyakov formula for the first variation of the Liouville action 

\begin{equation}\label{varaccessory}
	- \frac{1}{2\pi} \partial_{p_i} S_{L} = \gamma_{i},
\end{equation}
where $\gamma_{i}$ are the accessory parameters \eq{stress}.

For the case of $n = 4$, the volume of the W-P metric over the moduli space is \cite{Zograf1998,Wolpert2005} (see \ref{AppVolume} for more details).

\begin{equation}\label{volumeWP}
	\int_{\mathcal{M}_{0,4}} \omega_{WP} = \pi^{2}.
\end{equation}






\subsubsection{Takhtajan-Zograf metric}
On a surface with cusp singularities, it was shown in \cite{Takhtajan1991} that there exists a K\"ahler metric, now known as the Takhtajan-Zograf (T-Z) metric, associated with each cusp singularity


\begin{equation} 
	\omega_{TZ, i} = \frac{i}{2}\int \mu(z,\bar{z}) \wedge \overline{\mu(z, \bar{z})} \, E_{i}(z, 2) dV_{z} ,
\end{equation}
where 
\begin{equation}
	E_{i}(z, s) = \sum_{\gamma \in \Gamma_{i}/\Gamma} {\rm Im}( \sigma_{i}^{-1} \gamma z)^{s},
\end{equation}
is the Eisenstein-Maass series associated with the cusp $x_{i} = w(p_i)$. Here, $\Gamma_{i} = \{ P_{i}, P_{i}^{2}, ... \}$ is the cyclic subgroup generated by the parabolic element of $\Gamma$ which has the cusp at $x_i$ as a fixed point. The fractional linear transformation $\sigma_{i} \in PSL(2, \mathbb{R})$ sends $x_{i}$ to infinity via $\sigma_{i}^{-1} x_{i} = \infty$, and similarly modifies the corresponding generator $\left(\sigma_{i}^{-1} P_{i} \sigma_{i}\right)(w)= w \pm 1$. In \cite{Park2015}, the K\"ahler potential for the T-Z metric was found to be related to the Fourier coefficients in \eq{KleinH}. Specifically, writing (note slightly different convention from \cite{Park2015})
\begin{eqnarray}\label{Hdef}
	H_{i} = |a_{i}(1)|^{2}, \quad i = 1, ..., n-1, \quad H_{n}  = 1/|a_{n}(-1)|^{2}.
\end{eqnarray}
Projected to the moduli space $\mathcal{M}_{0,n}$ (\ref{moduli_p}), the T-Z metric of cusp $x_{i}$ is 

\begin{equation}
\omega_{TZ, i}  = - \frac{3}{8\pi} \sum_{j,k} \frac{\partial^{2} \log H_{i}}{\partial p_{i} \partial \bar{p}_{j}} \, \frac{i}{2} dp_{i} \wedge d\bar{p}_{j}= \frac{3}{8\pi} i \partial_{\bar{p}} \partial_{p} \log H_{i}.
\end{equation}
Summing over all cusps and writing $\omega_{TZ} = \sum_{i = 1}^{n} \omega_{TZ,i}$, and $H = \prod_{i = 1}^{n} H_{i}$, one gets

\begin{equation}
 \omega_{TZ} = \frac{3}{8\pi} i \partial_{\bar{p}} \partial_{p} \log H .
\end{equation}
For the case of $n = 4$, the volume of the T-Z metric is \cite{Wolpert2005} 

\begin{equation}\label{volumeTZ}
	\int_{\mathcal{M}_{0,4}}  \omega_{TZ} = 3.
\end{equation}
We give a physical derivation of this result in \ref{AppVolume}.

\subsubsection{Global K\"ahler potential}

In \cite{Park2015}, a global K\"ahler potential $\mathscr{S}_{L}$ was found for a special combination of the W-P and T-Z metrics
\begin{equation}\label{SLmod}
	\mathscr{S}_{L} = S_{L} - \pi  \log H,
\end{equation}
which may be thought of as a modification of the Liouville action for a metric with cusp singularities. The function $\mathscr{S}_{L}$ is the Quillen metric on the determinant line bundle on $\mathcal{M}_{0,n}$ \cite{Park2015}. The first variation produces
\begin{eqnarray}\label{SLmodvar}
	\partial_{p_i} \mathscr{S}_{L}
	 = - 2\pi \gamma_{i} - \pi \sum_{k} \partial_{p_i} \log H_{k},
\end{eqnarray}
and the K\"ahler form is

\begin{equation}
	i \partial_{\bar{p}} \partial_{p} \mathscr{S}_{L} = \omega_{WP} -\frac{4\pi^{2}}{3} \omega_{TZ}.
\end{equation}
For $n = 4$, the volume of this special combination follows \eq{volumeWP} and \eq{volumeTZ}

\begin{equation}
	\int_{\mathcal{M}_{0,4}} \left( \omega_{WP} - \frac{4\pi^{2}}{3} \omega_{TZ}\right) =-3\pi^{2}.
\end{equation}

\section{Holomorphic Quantum Hall States on Riemann Surfaces}\label{physics}

This section provides the physics background, where we discuss in more detail the physical motivation for studying the generating functional \eq{Z}. We define the quantum Hall model wave functions on genus-0 Riemann surfaces. With the exception of (\ref{holopsi}), this material is mostly review of Ref. \cite{Can2015}. A recent comprehensive review of QH states on Riemann surfaces can be found in the lecture notes \cite{Klevtsov2016a}. 

\subsection{Lowest Landau Level on Hyperbolic surfaces}
We consider model wave functions of $N$ coordinates $z_i$ which are zero modes of the anti-holomorphic kinetic momentum operator \cite{Iengo1994}

\begin{equation}\label{LLLconstraint}
	D_{\bar{z}_{i}} \Psi(z_1, ..., z_N) = 0, \quad {\rm where} \quad  D_{\bar{z}} = \partial_{\bar{z}} - i \frac{e}{\hbar} A_{\bar{z}} + i s \omega_{\bar{z}}.
\end{equation}
We refer to this as the lowest Landau level (LLL) constraint. Here, $A_{\bar{z}}$ and $\omega_{\bar{z}}$ is anti-holomorphic component of the gauge and spin connection, respectively. The curvature for these connections are the magnetic field $e^{-\phi} 2i(\partial_{\bar{z}} A_{z} - \partial_{z} A_{\bar{z}}) = B $ and the scalar curvature $2i e^{-\phi} (\partial_{\bar{z}} \omega_{z} - \partial_{z}\omega_{\bar{z}}) = \frac{1}{2} R$. In the language of differential forms, the magnetic field strength is $F = d A $ where $F = B dV$ and $dV = e^{\phi} i dz \wedge d\bar{z}/2$, while the Ricci tensor ${\rm Ric} = d\omega $,  where ${\rm Ric} = (R/2)\, dV$ and $R = - \Delta \phi$ is the scalar curvature. For reference, the Gauss-Bonnet theorem reads $\int {\rm Ric} = 2\pi \chi(\Sigma)$. The anti-holomorphic components are built from the Cartesian components by $A_{\bar{z}} = \frac{1}{2} (A_{1} + i A_{2})$, and similarly for the spin connection. In transverse gauge, $A_{z} = (i \hbar/2e) \partial_{z} Q$,  $\omega_{\bar{z}} = -(i/2)\partial_{\bar{z}} \phi$, $A_{z} = A_{\bar{z}}^{*}$ and $\omega_{z} = \omega_{\bar{z}}^{*}$. The deformation parameter $s$ appears as a coupling constant for the spin connection, and for this reason it is referred to as spin. Below we set $e = \hbar = 1$, so that the flux quantum $h/e = 2\pi$. 

Being on a closed hyperbolic surface imposes a constraint not just on the total flux, but on the magnetic field itself. Let us see how this works. Dirac quantization requires 

\begin{equation}
	N_{\phi} = \frac{1}{2\pi} \int_{\Sigma} F  = \frac{B V}{2\pi} \in \mathbb{N}
\end{equation}
whereas the Gauss-Bonnet theorem implies that the volume $VR_{0} =   -4\pi (n - 2 + 2g )$. Combining these gives the commensurability condition on the ratio $B/R_{0}$ \eq{defk1} which was found in \cite{Avron1992,Pnueli1994b}

\begin{equation}\label{defk}
	 k \equiv -\frac{4B}{R_{0}} = \frac{B\lambda}{2}  = \frac{2N_{\phi}}{(n - 2 + 2g)}
\end{equation}
For a scalar curvature $R_{0}  = - 8/\lambda$, this implies the magnetic field is a rational number in units of $1/\lambda$, 

\begin{equation}
	B =  \frac{4N_{\phi}}{( n - 2 +2g)} \frac{1}{\lambda}
\end{equation}


If there are fluxes $a_i$ threading the punctures, then the Dirac condition requires the integrality of $N_{\phi} = B V / 2\pi - \sum_{i = 1}^{n} a_{i}$ and $B = 4 (N_{\phi} + \sum_{i} a_i ) / (n - 2 + 2g) \lambda$.

On a hyperbolic surface, the energy spectrum of the Hamiltonian $H = D_{\bar{z}}^{\dagger} D_{\bar{z}}$	is rather complicated \cite{Avron1992}, and involves discrete as well as continuous bands. However, at the bottom of the spectrum there are still discrete, degenerate Landau levels. The lowest Landau level occurs at zero energy and is spanned by the zero modes of \eq{LLLconstraint}, consisting of holomorphic sections $s_{n}(z)$ normalizable under the inner product \cite{Klevtsov2013}

\begin{equation}
	\psi_{n}(z) = s_{n}(z) e^{\frac{1}{2} (Q - s \phi )}, \quad || s_{n}||^{2} = \int |s_{n}(z)|^{2} e^{ Q - s \phi} dV.
\end{equation}
The degeneracy is given by the Riemann-Roch formula \cite{Levay1999b,Avron1992}
\begin{eqnarray}\label{degLLL}
	N = N_{\phi} + (1 - 2s) (1 - g), \quad 0<s \le 1\label{deg1}\\
	N = N_{\phi} + (1-2s) (1 - g)  + n (\lceil s \rceil  - 1) , \quad s >  1\label{deg2}
\end{eqnarray}
where $\lceil s \rceil$ is the smallest integer larger that $s$. For genus $g = 0$, this can be seen as a consequence of normalizability. Close to the punctures 
\begin{eqnarray}\label{psipuncture}
	|\psi_{n}|^{2} e^{ \phi} d^{2} z  
	&=|s_n(z)|^{2}|z-p_i|^{2(s - 1)} \left| \log |z - p_i|^{2}\right|^{k + 2s - 2} d^{2} z.
\end{eqnarray}
Since the sections behave as $|s_{n}(z)|^{2} \sim |z - p_i|^{2l_{i}}$, normalizability requires $l_{i} + s > 0$. Therefore for $i \ne n$, we have holomorphic sections with poles of order $l_i  = 1 - \lceil s \rceil, -\lceil s\rceil , ..., -1$ that are linearly independent. For $0<s \le  1$, $\lceil s \rceil = 1$ and there are no additional states at the punctures. For $s >1$, the multiplicity of these sections is $(n-1) (\lceil s \rceil  - 1)$. 
At infinity, using the limits \eq{cusp} and \eq{kahler}
\begin{equation}
|\psi_{n}(z)|^{2} e^{\phi}d^{2}z \sim |z|^{2n} |z|^{-2N_{\phi}} |z|^{2(s-1)} \left( \log |z|^{2}\right)^{N_{\phi} + 2s - 2} d^{2}z.
\end{equation}
In local coordinates $\xi = 1/z$ around infinity, 
\begin{eqnarray}
|\psi_{n}(\xi)|^{2} e^{\phi}d^{2}\xi 
&   \sim |\xi|^{2(N_{\phi} - n - s)} |\xi|^{-2} \left( \log |\xi|^{2}\right)^{N_{\phi} + 2s - 2}.
\end{eqnarray}
Integrability at infinity then requires $N_{\phi} - s - n > 0$, which sets an upper bound on the degree of $s_{n}(z)$ to $n_{max} = N_{\phi} - s -1$. The total number of degenerate states, and hence the total number of particles in the many-body ground state $N$, is just $N = N_{\phi} - s + (n-1)(s-1) = N_{\phi}+1-2s + n(s-1)$, which agrees with \eq{deg2} at $g = 0$ for $s > 1$. For $0<s \le 1$, the condition is $N_{\phi}-2s  -  n_{max} =0$, and the degeneracy is \eq{deg1}. 

We will consider positive $s \le 1$ in what follows. In this case, the holomorphic sections will consist of polynomials in $z$ of degree at most $n_{max}$ and at least $0$, i.e. $s_{n}(z) = \sum_{0\le k \le n_{max}} c_{n}^{(k)} z^{k}$.




The many-electron wave function for a completely filled LLL is the Slater determinant \cite{Klevtsov2013}

\begin{equation}
	\Psi = \frac{1}{\sqrt{N!}}{\rm det}(s_{i}(z_{j})) \exp \left( \frac{1}{2}\sum_{i} (Q(z_i) - s \phi(z_i)\right). 
\end{equation}
If we do not include sections with poles at the punctures, then under a change of basis the wave function reads

\begin{equation}\label{psiIQH}
	\Psi = \frac{1}{\sqrt{\mathcal{Z}}} \prod_{i<j}(z_i - z_j) e^{ \frac{1}{2} \sum_{i} (Q(z_i) - s \phi(z_i))},
\end{equation}
where $\mathcal{Z}$ is the normalization integral. This is the free electron ground state which exhibits the integer quantum Hall (IQH) effect at filling fraction $\nu = 1$. 

\subsection{FQH model wave functions}
We consider model wave functions of the form

\begin{equation}\label{Psigen}
	\Psi = \frac{1}{\sqrt{\mathcal{Z}[\{p_k\}]}} F(z_{1}, ..., z_{N}) \prod_{i,k}(z_i - p_k)^{a_{k}}\exp \left[\frac{1}{2} \sum_{i = 1}^{N}  \left(Q(z_i, \bar{z}_{i})  - s \phi(z_{i}, \bar{z}_{i})\right) \right],
\end{equation}
where $Q$ is the magnetic potential satisfying $- \Delta Q= 2  B$. This wave function consists of three parts. First, the single-particle factor $e^{\frac{1}{2} (Q - s \phi)}$ which all multi-particle states have in common, is inherited from the LLL constraint \eq{LLLconstraint}.  

The symmetric polynomial $\prod_{i,k}(z_i - p_k)^{a_k}$ can be understood as merely adding zeros at $p_k$ to the wave function to ensure integrability around the punctures (see below). However, the physical interpretation of such a term is that it represents fluxes threading the punctures, and the wave function is written in a gauge (requiring a branch cut for non-integer $a_i$) for which the Aharonov-Bohm phase $\exp(i 2\pi \alpha_{k})$ around the fluxes is explicitly captured by monodromy. For integer $a_{k}$, the flux tube can be interpreted as a Laughlin quasi-hole.

The defining features of the FQH state are encoded in the holomorphic polynomial $F$. Under a M\"obius transformation $f(z) = \frac{a z + b}{c z + d}$ with $a d - bc = 1$, the holomorphic polynomial transforms as

\begin{equation}
	F(f(z_1), ...., f(z_{N})) = \prod_{i = 1}^{N} \left(f'(z_{i}) \right)^{\nu^{-1} N - \mathcal{S}} F(z_{1}, ..., z_{N}).
\end{equation}
In \cite{Can2015}, we referred to $-\frac{1}{2}( \nu^{-1} N + \mathcal{S})$ as the holomorphic dimension of the wave function. That the dimension scales at most linearly with $N$ is a necessary condition for the polynomial to be an admissible QH state. Furthermore, the dimension can be used as the definition of the filling factor $\nu$ and the shift $\mathcal{S}$ \cite{Wen1992b, Read2011}. To understand why, consider the round metric $\phi = - 2 \log \left( 1 + \pi |z|^{2}/V\right)$ on a sphere with a constant magnetic field. In this case, $Q = -N_{\phi} \log \left( 1 + \pi |z|^{2}/V\right)$, and the isometry group of the round metric $SU(2)$ is the subset of global conformal transformations which take $z \to (a z + b)(- \bar{b}z + \bar{a})$ with $|a|^{2}+ |b|^{2} = 1$. Thus, the wave function will transform under $SU(2)$ as

\begin{equation}\label{psi_trans}
	| \Psi ( f(z_{1}), ..., f(z_{N}) )|^{2} = \prod_{i = 1}^{N} | f'(z_{i})|^{\nu^{-1}N - \mathcal{S} - N_{\phi} + 2s} | \Psi (z_{1}, ..., z_{N})|^{2}
\end{equation}
Invariance under $SU(2)$ is equivalent to asking that the mean density be constant on a sphere, a necessary condition for incompressibility. According to (\ref{psi_trans}), this is achieved when the number of particles is chosen such that

\begin{equation}
	N =  \nu N_{\phi} + \nu (\mathcal{S} - 2s)
\end{equation}
Comparing with (\ref{Nmax}), we can easily identify the geometric response coefficient $\mu_{H}$ as

\begin{equation}\label{mushift}
\mu_{H} = \frac{1}{2} \nu (\mathcal{S} - 2s).
\end{equation}
We see from these considerations that $\mu_{H}$ is determined by a combination of the holomorphic dimension of the wave function and the spin $s$. 

We can apply this reasoning to the state with fluxes $a_{i}$ \eq{Psigen}. In this case, the coordinate transformation will also move the positions of the fluxes, giving us a covariant transformation law

\begin{equation}
	| \Psi ( f) |^{2} = \prod_{i = 1}^{n} |f'(p_i)|^{N a_{i}} \prod_{i = 1}^{N} | f'(z_i)|^{\nu^{-1} N - \mathcal{S} - N_{\phi} + 2s + \sum_{i} a_{i}} |\Psi(z)|^{2},
\end{equation}
from which we see that
\begin{equation}
	N = \nu N_{\phi} - \sum_{i} \nu a_{i} + 2 \mu_{H}
\end{equation}
and that the flux insertions transform with a conformal dimension $ =  -N a_{i}/2$. This dimension encodes the Berry phase of the adiabatic transport of quasi-holes, and we may unpackage it in the following way (see e.g. \cite{Can2015,Li1992, Gromov2016}) 
\begin{equation}
	\gamma_{a_i} = -2\pi a_{i} N =  -2\pi \nu a_{i} N_{\phi}    -4\pi \times \frac{a_{i}}{2} (2 \mu_{H}- \nu a_{i}) +2\pi \sum_{j \ne i} \nu a_{i} a_{j}.
\end{equation}

The first term $-2\pi \nu a_{i} N_{\phi} = -\nu a_{i} \int B dV$ is the Aharanov-Bohm phase and implies the charge $-\nu a_i$ \cite{Laughlin1983}. The second $h_{a_{i}} \frac{1}{2} \int R dV$ is the spin-curvature A-B phase and implies the spin $h_{a_i} = \frac{a_i}{2}(2\mu_{H} - \nu a_{i})$. And the last term is the statistics, and implies the mutual exchange statistics $\pi \nu a_{i} a_{j}$ between $i$ and $j$ \cite{Arovas1984}. 

In this paper, we focus on the Laughlin wave function for which

\begin{equation}
	F = \prod_{i<j} (z_{i} - z_{j})^{\beta}, \quad \nu = \beta^{-1}, \quad \mathcal{S} = \beta
\end{equation}
Note that $\beta = 1$ describes the completely filled LLL \eq{psiIQH}. The density matrix of this state is given by the Bergman kernel, and the mean density of electrons is given by the Bergman kernel on the diagonal \cite{Douglas2010, Klevtsov2013} (see Sec. (\ref{IQHexact})).

Since we consider constant magnetic field on a hyperbolic sphere, we will use for the metric

\begin{equation}
	\phi = \log |w'(z)|^{2} - 2 \log\Big( {\rm Im}[ w(z)] \Big),
\end{equation}
and for the magnetic potential $Q = - B K / 2$, or explicitly
\begin{equation}\label{Q}
Q =  k \log\Big( {\rm Im}[ w(z)] \Big),
\end{equation}
However, the wave function \eq{Psigen} with the potential \eq{Q} is not normalizable for all values of $s$ and $a_{i}$ under the standard inner product $||\Psi||^{2} = \int \prod_{i} dV_{i} |\Psi|^{2}$. In the neighborhood of a puncture, the integrand behaves as

\begin{equation} |z_1 - p_i|^{2(a_i+s-1)}\left(\log|z_{1}-p_i| \right)^{ (k + 2s - 2)}
\end{equation}
which constrains $a_{i} + s > 0$ for all $i$. This condition ensures that the wave function vanishes at the punctures, and is thus integrable in the neighborhood of $p_i$.
 Thus, the unnormalized Laughlin wave function on a punctured sphere has the form

\begin{equation}
	\Psi_{L} = \prod_{i = 1}^{N} \prod_{j = 1}^{n} (z_i - p_j)^{a_j}\prod_{i<j}(z_{i} - z_{j})^{\beta} \prod_{i = 1}^{N} |w'(z_{i})|^{-s}\left( {\rm Im}[ w(z_i) ]\right)^{\frac{1}{2}(k + 2s)}
\end{equation}
The explicit form of the metric is generally not available, but its asymptotic behavior near the cusp singularities is known. This turns out to be sufficient for our purposes later on.

\subsection{Wave function on the fundamental domain}\label{holopsi}
On smooth surfaces, a consequence of the holomorphicity of the wave function is that, aside from the normalization factor, it will have a holomorphic dependence on the complex structure. Since the puncture coordinates are essentially a parameterization of the complex structure, we expect there is a way to write the wave functions as holomorphic being $p_i$. Here, we argue that this is accomplished by going to the fundamental domain.

The Klein map $J(w)$ provides us with a way to parameterize the Laughlin wave function by coordinates on $\mathbb{H}/\Gamma$. The electron coordinates $z_i = J(w_i)$, where $w_i = w(z_i)$ is a point in the fundamental domain, and $x_i = w(p_i)\in \mathbb{R}$ are the cusps. The wave function should be properly expressed in terms of $s-$differentials

\begin{equation}
	\Psi_{L} = \prod_{i,j}(J(w_i) - x_j)^{2a_j}\prod_{i<j} \left(J(w_i) - J(w_j)\right)^{\beta} \prod_{i = 1}^{N} e^{\frac{1}{2} Q(J(w_i))} \left( \frac{dw_{i}}{{\rm Im}\, w_i} \right)^{- s}
\end{equation}
Under elements of $\gamma \in \Gamma$, the Klein map is strictly invariant $J(\gamma w) = J(w)$, and the wave function will transform by a phase \cite{Pnueli1994b}

\begin{equation}
	\Psi_{L}(\{\gamma w_i\}) = \prod_{i = 1}^{N}  \left( \frac{c w_i + d}{c \bar{w}_i + d} \right)^{\frac{1}{2}(k + 2s) } \Psi_{L}(\{w_{i}\}).
\end{equation}
Thus, the probability density is invariant under the action of $\Gamma$. Since they are related by a single-valued map $J(w)$, the normalization of this wave function on $\mathbb{H}/\Gamma$ is identical to the normalization on $\Sigma$. 

In this picture, the dependence on moduli is entirely contained in the dependence on the analytic map $J$, which is holomorphic in the cusp coordinates $p_i$. It is this property which is necessary for the Berry curvature to be K\"ahler (see Sec. (\ref{BerryCurvatureProof})). 

\section{Ward Identity}\label{WardIdentity}

The Ward identity (also called the loop equation) \cite{Zabrodin2006} provides a powerful approach to computing correlation functions and extracting sum rules for the Laughlin wave function. The identity itself is a trivial statement which asserts the vanishing of the integral of a total derivative

\begin{equation}
\int \prod_{k = 1}^{N} d^{2} z_{k} \, \sum_{i = 1}^{N} \frac{\partial}{\partial z_{i}} \left( \frac{1}{z - z_{i}} \prod_{i<j}|z_{i} - z_{j}|^{2\beta} \prod_{k} e^{ W(z_k , \bar{z}_k) + Q_{a}(z_k, \bar{z}_k)}  	\right) = 0.
\end{equation}
where $W = - B K / 2 + (1 - s) \phi $, and

$$Q_{a} = \sum_{i = 1}^{N} \sum_{k = 1}^{n} a_{k} \log |z_i - p_k|^{2}.$$
 Introducing the potential for the density

\begin{equation}\label{potrho}
	\varphi(z, \bar{z}) = -\beta \sum_{i} \log|z - z_{i}|^{2}, \quad - \Delta \varphi = 4\pi \beta \rho(z, \bar{z})
\end{equation}
We may write the Ward identity in integro-differential form

\begin{equation}\label{WI1}
	 \int \frac{\partial_{\xi} (W + Q_{a})}{z - \xi} \langle \rho(\xi)\rangle dV_{\xi} = - \frac{1}{2\beta}\langle (\partial_{z} \varphi(z))^{2}\rangle  + \frac{( \beta-2)}{2\beta} \langle \partial_{z}^{2} \varphi(z) \rangle
\end{equation}

This identity does not give us something for nothing. Notice the first term on the RHS is a two-point function evaluated at coincident points, whereas the rest of the identity involves one-point functions. This is the start of the famous BBGKY hierarchy \cite{Feynman}, and if we ever hope to extract something useful we must truncate the hierarchy and turn this into an equation that can be solved. It was argued in \cite{Zabrodin2006,Can2014, Can2015} that the connected two-point function of the potential field $\varphi$ is given by

\begin{equation}
	\langle \varphi(z) \varphi(z') \rangle_{c} = 4\pi \beta G(z, z')
\end{equation}
where $G(z, z')$ is the Green's function for the Laplacian $- \Delta G = \delta - 1/V$ which satisfies $\int G dV = 0$. From this, a crucial step, supported by all the available evidence but which remains a conjecture, is that this correlation function at coincident points must be regularized, and that the correct (and most natural) regularization is

\begin{equation}\label{UVreg}
\lim_{z \to z'}\langle \varphi(z) \varphi(z') \rangle_{c}	 = 4\pi \beta \lim_{z \to z'} \left( G(z, z') + 2\pi \log d(z, z')\right)
\end{equation}
where $d(z, z') = \int_{z}^{z'
} e^{\phi/2} |dz|$  is the geodesic distance between the two-points. According to this prescription,

\begin{equation}\label{schwarz}
	\langle (\partial_{z} \varphi(z))^{2} \rangle_{c} = 4\pi \beta \lim_{z \to z'}\partial_{z} \partial_{z'} \left( G(z, z') + 2\pi \log d(z, z')\right) = \frac{\beta}{6} \mathcal{S}[\phi].
\end{equation}
where $\mathcal{S}[\phi] = \{w(z), z\}$ is the Schwarzian of the metric (\ref{schwarz1}) discussed in Sec.(\ref{Geometry}). In this way, we are able to close the hierarchy and transform the Ward identity into a integro-differential equation for the one-point function. 

We wish to develop a large $k$ (equivalently large $N$) expansion for the density. This must be done with great care in the presence of cusps. Motivated by the analogy to the 2D Coulomb plasma, we hypothesize that the curvature singularities behave as a sort of impurity. In the screening phase, the effects of this impurity on local observables such as the density will not be long-ranged. Importantly, this does not preclude some long-range effects on the potential $\varphi$. Given this, we assume that the density has the following asymptotic expansion 

\begin{equation}\label{rhoasymp}
	\langle \rho \rangle = \bar{\rho} + \sum_{i=1}^{n} \rho_{i}  
\end{equation}
where $\bar\rho$ has a smooth expansion away from the cusps

\begin{equation}\label{smooth_exp}
	\bar{\rho} = k a_{1} + a_{0} + k^{-1} a_{-1} + ....
\end{equation}
and the density $\rho_{i}$ has weight concentrated near the cusp, and is exponentially suppressed away from $p_i$ on a scale set by the magnetic length, or the Debye length in the plasma picture. This means sufficiently far from the punctures, the density is $\langle \rho \rangle  =   \bar{\rho}$ to leading order in large $k$. Furthermore, the cusp density $\rho_{i}$ has support only in a vanishingly small neighborhood of $p_i$ in the large $k$ limit. 

This hypothesis is supported by rigorous results for the special case of $\beta = 1$, $s = 0$, and $a_k = 0$, for which $\langle \rho \rangle$ is the Bergman kernel on the diagonal \cite{Auvray2016, Sun2016}. Results on the density were also obtained for conical singularities on orbifolds (which constrains the possible opening angles) in \cite{Ross2009}, as well as conical singularities with arbitrary opening angle in \cite{Ameur2016,Ameur2016a}. In this context, the screening behavior is referred to as localization of the Bergman kernel, since only local information (about the holomorphic sections) is needed to describe the density in regions far from the cusps \cite{Auvray2016}. From the Coulomb plasma perspective, $\beta= 1$ is not a special value, and it is reasonable to assume that the qualitative features of the density such as its screening behavior, should be preserved for $\beta \ne 1$.

\subsection{Smooth expansion away from punctures}

To get a handle on the smooth part of the expansion away from the singularities, we consider $z \in \Sigma$ sufficiently far from all the punctures, and apply the covariant anti-holomorphic derivative twice. In this region, $\langle \rho \rangle$ is nowhere vanishing, and we can rewrite the Ward identity using \eq{schwarz} as a nonlinear differential equation


\begin{equation}\label{ward_diff}
	 \langle \rho \rangle \approx   \frac{ B}{2\pi \beta}   + \frac{(\beta - 2 s)}{8\pi \beta} R + \frac{(2 - \beta)}{8\pi \beta} \Delta \log \langle \rho \rangle  - \frac{1}{12\pi^{2} \beta}e^{ - \phi} \partial_{\bar{z}} \left(  \frac{\partial_{\bar{z}} \mathcal{S}[\phi]}{e^{\phi} \langle \rho \rangle} \right),
\end{equation}
which clearly reveals the structure of the expansion \eq{smooth_exp}. Since the magnetic field and curvature are constant away from the punctures, the asymptotic expansion truncates and the solution for the density is simply 

\begin{equation}\label{rhobarWI}
	\bar{\rho} = \frac{\nu B}{2\pi } + \frac{\mu_{H}}{4\pi} R_{0} = \frac{2}{\lambda \pi} \left(  \nu k- \mu_{H}\right)
\end{equation}
where we identify $\nu = \beta^{-1}$ and $\mu_{H} = \frac{1}{2}\nu (\beta  - 2  s)$ in \eq{ward_diff}.

\subsection{Charge Moment}
Integrating \eq{rhobarWI} over the entire surface gives

\begin{equation}
	\int_{\Sigma} \bar{\rho} dV = \frac{ \nu B V}{2\pi } + \frac{\mu_{H}}{4\pi} R V = \nu N_{\phi} + \mu_{H} (2-n ).
\end{equation}
Comparing this to the formula for $N$ \eq{Nmax}, implies 

\begin{equation}
	\int_{\Sigma} \left( \langle \rho \rangle - \bar{\rho}\right) = n \mu_{H} - \sum_{i} \nu a_{i}.
\end{equation}
Since the integrand has support only along the divisor $X$, we may rewrite the LHS as

\begin{equation}
	\sum_{i = 1}^{n} \lim_{\epsilon \to 0}\int_{D_{\epsilon}(p_{i})} \rho_{i}\, dV  = n \mu_{H} - \sum_{i} \nu a_{i}.
\end{equation}
If all the fluxes are equal, the manifest symmetry of the configuration of punctures would imply that any excess charge is shared equally among all the punctures. Otherwise, by merely relabeling punctures (i.e. doing nothing) we could transfer charge. If then one flux is slightly adjusted, the screening property would imply that this procedure is invisible to the other punctures, and only affects the local density. 
Therefore, we conclude that the local charge moment is
\begin{equation}\label{localcharge}
\lim_{\epsilon \to 0}\int_{D_{\epsilon}(p_{i})} \, \rho_{i} \, dV = \mu_{H} - \nu a_{i}
\end{equation}
The charge moment also follows from the leading asymptotic expansion of the potential \eq{potrho} 
\begin{equation}\label{phiasymp}
	\langle \varphi \rangle = -\frac{B}{2} K + Q_{a}  + \beta \mu_{H} \phi + \tilde{\varphi}
\end{equation}
which has the correct pole behavior near the punctures

\begin{equation*}
	\langle \varphi \rangle \sim - \beta \frac{( \mu_{H}- \nu a_{i})}{z- p_i},
\end{equation*}
to reproduce \eq{localcharge}. Here, the correction term $\tilde{\varphi} = O(k^{-1})$, and as we argue below, should consist of functions highly concentrated near the singularities.

\subsection{Screening of Potential} \label{sectionCWI}
The screening hypothesis for the Laughlin state, encapsulated in \eq{rhoasymp}, does not imply that $\langle \varphi \rangle$ is also exponentially suppressed away from ``impurities". In particular, the leading asymptotics \eq{phiasymp} imply a roughly logarithmic {\it growth} in the neighborhood of a singularity. What can we say, however, about $\tilde{\varphi}$ which enters at higher orders? This can be seen from another rewriting of the Ward identity. Applying just a single anti-holomorphic derivative to eliminate the integral in \eq{WI1}, and using \eq{phiasymp}, we find that away from a singularity where $\langle \rho \rangle >0$, the subleading correction $\tilde{\varphi}$ satisfies




\begin{equation}
	\partial_{z} \tilde{\varphi} = \frac{(2 - \beta)}{2} \partial_{z} \log \langle \rho \rangle + \frac{1}{2\pi \beta \langle \rho \rangle } e^{ - \phi} \partial_{\bar{z}} \langle (\partial \varphi)^{2} \rangle_{c}
\end{equation}
In this form, it is clear that $\partial_{z} \tilde{\varphi}$ is suppressed away from the punctures. This follows because $\log \langle \rho \rangle \sim \log \rho_{0} + \rho_{i} / \rho_{0}$ which is a constant plus a rapidly decaying term. Similarly, the second term involves to leading order the Schwarzian of the metric according to \eq{schwarz}. Acting on the Schwarzian by the antiholomorphic derivative will result in $\partial_{\bar{z}} \mathcal{S}[\phi]  = \pi \sum_{i}\left( -\partial_{z}  + \gamma_{i} \right)\delta(z - p_i)$, which is also clearly concentrated near the puncture.

\subsection{Sum Rules are Residues of the Ward Identity}

Before diving into the technical details of the derivation of the sum rules mentioned in the introduction, we will give a brief outline of our plan of attack. The first step is to rework the Ward identity to the form

\begin{eqnarray}\label{WI2}
\int \frac{\partial_{\xi} W}{z - \xi} \left( \langle \rho (\xi ) \rangle - \bar\rho \right) dV_{\xi}  = \mathcal{T}(z).
\end{eqnarray}
Here, $\mathcal{T}(z)$ acts as a generating function for the sum rules, as we will now show. Using \eq{rhoasymp}, we can rewrite the integral as

\begin{equation}
	\mathcal{T}(z) =  \sum_{i=1}^{n-1} \mathcal{T}_{i}(z), \quad \mathcal{T}_{i}(z) = \int \frac{\partial_{\xi} W}{z - \xi}\,  \rho_{i}(\xi) dV_{\xi}
\end{equation}
where the behavior at $p_n$ must be treated separately (see \ref{App_inf}). Under the hypothesis of the previous section, the integrand in $\mathcal{T}_{i}$ is highly concentrated near the puncture $p_{i}$, so that we may develop a series expansion of the integrand around $z = p_i$ 

\begin{eqnarray*}
\mathcal{T}_{i}(z) \approx  \frac{1}{(z - p_i)} \int_{|\xi - p_i| < |\zeta_i| }\partial_{\xi} W \, \rho_{i} dV_{\xi} + \frac{1}{(z - p_i)^{2}} \int_{|\xi - p_i| < |\zeta_i|} (\xi - p_i) \partial_{\xi} W \rho_{i} dV_{\xi} + ...
\end{eqnarray*}

Although this expansion will also include terms like $1/\zeta_i^{n} \log^{m}|\zeta_i |$ (as subleading corrections from the limits of integration), we are only interested in the meromorphic terms since these encode the sum rules. Namely, the translation sum rule is the residue of the simple pole 

\begin{eqnarray}\label{translateT}
	\mathcal{L}_{-1}(p_i) =  \lim_{\epsilon \to 0} \frac{1}{2\pi i}\oint_{C_{\epsilon}(p_i)} \mathcal{T}(z)\, dz  = \int \partial_{\xi} W \left( \langle \rho (\xi) \rangle - \bar\rho \right) dV_{\xi}
\end{eqnarray}
where $C_{\epsilon}(p_i)$ is the contour $p_i + \epsilon e^{i \theta}$, $\theta \in [0,2\pi)$. The dilatation sum rule is the residue of the second order pole
 
 \begin{equation}\label{dilateT}
 \mathcal{L}_{0}(p_i) = 	\lim_{\epsilon \to 0}\frac{1}{2\pi i}\oint_{C_{\epsilon}(p_i)} (z - p_i) \mathcal{T}(z) dz  = \int (\xi - p_i) \partial_{\xi} W \, \left( \langle \rho (\xi) \rangle - \bar\rho \right) dV_{\xi}.
\end{equation}
These results imply that, for these sum rules to be nonzero, the generating function $\mathcal{T}_{i}(z)$ must exhibit {\it algebraic} decay away from the puncture. This structure that emerges can be tied to a conformal symmetry arising due to the presence of singularities \cite{Laskin2016}. Indeed, without the punctures, $\mathcal{T}_{i}$ would not exhibit such behavior. The behavior at infinity is a bit subtle and is treated in \ref{App_inf}.  

\subsection{Computing Sum Rules from iterative solution of Ward identity}\label{sum_from_res2}
This section contains a detailed technical derivation of the leading contribution to $\mathcal{T}(z)$ in the limit of large $N$. The main result is \eq{T_poles}-\eq{T_poles2} which shows the first and second order poles of $\mathcal{T}(z)$ near the cusp singularities.  


In order to get \eq{WI1} into the form \eq{WI2}, we must subtract from both sides the expression

\begin{eqnarray*}
	\int \frac{\partial_{\xi}W}{z -\xi} \bar\rho dV_{\xi} &= \pi \bar\rho \left( - \frac{B}{4}  ( \partial_{z}K(z))^{2} + (1 - s) \partial_{z}^{2} K \right),\\
\end{eqnarray*}

and 
\begin{eqnarray*}
\int \frac{\partial_{z} Q_{a}}{z - \xi} \langle \rho(\xi)\rangle dV_{\xi}	 & = \frac{1}{\beta} \sum_{i} \frac{ a_{i}}{z - p_{i}} \Big\langle  (\partial_{z}\varphi) (p_{i}) -  \partial_{z}\varphi(z)\Big\rangle\\
& = - \frac{1}{\beta} \partial_{z} Q_{a} \langle \partial_{z} \varphi(z) \rangle + \frac{1}{\beta}\sum_{i} \frac{ a_{i}}{z - p_{i}} \langle  (\partial_{z}\varphi)(p_{i})\rangle ,
\end{eqnarray*}
where $\langle (\partial_{z}\varphi)(p_i) \rangle = \lim_{z \to p_i} \langle \partial_{z}\varphi(z) \rangle$. This limit is rather subtle, but we will ultimately not need it explicitly to prove our main result. 

Utilizing the UV regularization \eq{UVreg}, the Ward identity now is written up to order $O(1)$
\begin{eqnarray*}
\int \frac{\partial_{\xi} W}{z - \xi} \left( \langle \rho (\xi) \rangle - \bar{\rho} \right) dV_{\xi} &= 	-\pi \bar{\rho} \left( - \frac{B}{4}  ( \partial_{z}K(z))^{2} + (1 - s) \partial_{z}^{2} K \right)\\
&  - \frac{1}{12} \mathcal{S}[\phi]  - \frac{1}{2\beta} \left( \langle\partial_{z} \varphi \rangle \right)^{2} -\frac{(2 - \beta)}{2\beta} \langle \partial_{z}^{2} \varphi \rangle\\
&  +  \beta^{-1} \partial_{z} Q_{a} \langle \partial_{z} \varphi(z) \rangle - \frac{1}{\beta}\sum_{i} \frac{ a_{i}}{z - p_{i}} \langle  (\partial_{z}\varphi)(p_{i})\rangle \\
& + O(k^{-1})
\end{eqnarray*}

Next, we must substitute the asymptotic expansion of $\langle \varphi \rangle$, which we also know up to order $O(1)$. The relevant terms are

\begin{eqnarray}
	- \frac{1}{2\beta} \left( \langle \partial_{z} \varphi \rangle \right)^{2} 
	& = - \frac{B^{2}}{8\beta} (\partial_{z} K)^{2} +  \frac{B}{2\beta} \partial_{z} K \left( \partial_{z} Q_{a} + \beta \mu_{H} \partial_{z}\phi \right)\\
	&- \frac{1}{2\beta} (\partial_{z}Q_{a})^{2} - \mu_{H} \partial_{z}Q_{a} \partial_{z}\phi - \frac{\beta \mu_{H}^{2}}{2} (\partial_{z} \phi)^{2} \\
	& + \frac{B}{2\beta} \partial_{z} K \partial_{z} \tilde{\varphi}  + O(k^{-1})\\
	- \frac{(2 - \beta)}{2 \beta} \partial_{z}^{2} \langle \varphi \rangle &= - \frac{(2 - \beta)}{2 \beta}  \partial_{z}^{2} \left( - \frac{B}{2} K + Q_{a} + \beta \mu_{H} \phi\right) + O(k^{-1}) \\
	 \partial_{z} Q_{a} \langle \partial_{z} \varphi(z) \rangle & =   \partial_{z} Q_{a} \left( - \frac{B}{2} \partial_{z} K + \partial_{z} Q_{a} + \beta \mu_{H} \partial_{z} \phi + \partial_{z} \tilde{\varphi}\right)
\end{eqnarray}

Combining these we obtain the generating function \eq{WI2} up to order one

\begin{eqnarray}\label{T1}
\mathcal{T}(z) 
	&= \frac{\mu_{H} B}{2} \left(  \frac{R_{0}}{8} (\partial_{z}K)^{2}   -  \nabla_{z}\partial_{z} K \right)
	 + (s-1) \mu_{H} \left(   \frac{ R_{0}}{4 } \partial_{z}^{2} K +\partial_{z}^{2}\phi \right)\label{T1line1}\\
	&  - \frac{c_H}{12} \mathcal{S}[\phi]+\frac{1}{2\beta} (\partial_{z} Q_{a})^{2}   + \frac{(\beta - 2 )}{2\beta} \partial_{z}^{2} Q_{a}\label{T1line2} \\
	 &- \frac{1}{\beta}\sum_{i} \frac{ a_{i}}{z - p_i} \langle (\partial_{z}\varphi)(p_i)\rangle  
	  + \frac{B}{2 \beta} \partial_{z} K \partial_{z} \tilde{\varphi} + O(k^{-1})\label{T1line3}
\end{eqnarray}
where we have made use of the covariant holomorphic derivative $\nabla_{z}= (\partial_{z} - \partial_{z}\phi )$. As we argued before, the difference $\langle \rho \rangle - \bar\rho$ has finite support along the divisor $X$. This should be reflected in the structure of the RHS. In particular, we take $z \to p_i$, and study each term. 

First, we have the leading term which scales with magnetic field
\begin{eqnarray*}
\frac{R_{0}}{8} (\partial_{z}K)^{2} - \nabla_{z}\partial_{z} K  = \frac{ \lambda \gamma_{i}}{\zeta_{i} \log |\zeta_{i}/a_{i}(1)|^{2}} + O(\log^{-2}|\zeta_i|), 
\end{eqnarray*}
The second term in the first line \eq{T1line1} is
\begin{eqnarray*}
	\partial_{z}^{2} \left( \frac{R_{0}}{4} K + \phi \right) & = \frac{1}{(z - p_i)^{2}} + O(1),
\end{eqnarray*}
and the contribution from the fluxes in line \eq{T1line2} near the punctures, 
\begin{eqnarray*}
(\partial_{z}Q_{a})^{2} + (\beta - 2) \partial_{z}^{2}Q_{a} = 	2\left(\sum_{i\ne j}\frac{a_i a_j}{p_i - p_j} \right)\frac{1}{z - p_i} +\frac{(a_i^{2} - a_i(\beta - 2))}{(z - p_i)^{2}} , 
\end{eqnarray*}

Finally, it remains to argue away the final term $ \partial_{z} K\partial_{z} \tilde{\varphi}$. As we saw in section (\ref{sectionCWI}), $\partial_{z} \tilde{\varphi}$ is concentrated near the cusp singularity, and will not contribute to either single or double poles.

Putting all of this together, we find that near the puncture point, the generating function is
\begin{eqnarray}\label{T_poles}
	\mathcal{T}_{i}(z) = &\left[(s - 1) \mu_{H}  - \frac{c_{H}}{24} + \frac{1}{2\beta} (a_i^{2} - a_i (\beta - 2)) \right] \frac{1}{(z - p_i)^{2}} \\
	& + \left[ - \frac{c_{H}}{12} \gamma_{i} -  \frac{a_{i}}{\beta} \langle (\partial_{z} \varphi) (p_i) \rangle + \frac{1}{\beta} \sum_{j \ne i} \frac{a_i a_j}{p_i - p_j} \right]\frac{1}{z - p_i}\label{T_poles2}\\
	& + O\left(\frac{1}{(z - p_i) \log|z-p_i|} \right) + O(k^{-1}).
\end{eqnarray}
%
Then according to \eq{translateT} and \eq{dilateT}, the sum rules are
\begin{eqnarray}
	\mathcal{L}_{-1}(p_i) = - \frac{c_{H}}{12} \gamma_{i} - \nu a_{i} \langle (\partial_{z} \varphi) (p_i) \rangle +  \sum_{j \ne i} \frac{\nu a_i a_j}{p_i - p_j} \label{sum_rulesT}\\
	\mathcal{L}_{0}(p_i) 
	 = - \frac{c_{H}}{24}   - \frac{1}{2} a_{i}(1 - 2\nu s - \nu a_{i})+ (s-1) (\mu_{H} - \nu a_{i}) \label{sum_rulesD}
\end{eqnarray}
for $a_{i} = 0$ we recover Prop.(\ref{sumrulesprop}). We discuss the behavior at infinity in \ref{App_inf}. 
\subsection{Moments from the Sum rules}\label{mom_from_sum}

The translation and dilatation sum rules can be written in terms of the moments defined in \eq{moments}. This can be seen by direct evaluation of the integrals involved in the sum rules. We expand the integrand in a series around the puncture point. Using the asymptotics (\ref{dcusprefined}-\ref{dcusprefined2}) and \eq{cusprefined}, we get

\begin{eqnarray}
	\partial_{z} W = \left( k + 2s - 2\right) \frac{1}{\log |\zeta_i / a_i(1)|^{2}} \left( \frac{1}{\zeta_i} + \gamma_{i} - \frac{2 {\rm Re}(\gamma_{i}\zeta_{i}) }{\zeta_{i}\log |\zeta_{i}|^{2}} + O(\zeta_{i}^{2})\right) \\
	+ (1 - s) \left( - \frac{1}{\zeta_{i}} + \gamma_{i} + O(\zeta_{i}^{2}) \right)\\
e^{\phi} = \frac{\lambda }{|\zeta_i |^{2} \log^{2}|\zeta_{i}|^{2}} \left( 1 +2 {\rm Re}( \gamma_{i} \zeta_{i}) - \frac{2{\rm Re}(\gamma_{i} \zeta_{i})}{\log|\zeta_{i}|} + O(\zeta_{i}^{2})\right)
\end{eqnarray}

Inserting this into \eq{dilateT} and assuming $\rho_{i}(|\zeta_{i}|)$ is radially symmetric in the neighborhood of the puncture, the angular integration will kill most terms, leaving for the dilatation sum rule
\begin{eqnarray}\label{dilate_moments}
	\mathcal{L}_{0} &=  \int_{D(p_i)} \left[ \frac{k+2s - 2}{\log|\zeta_{i}/a_{i}(1)|^{2}} + (s-1)\right] \, \rho_{i} \frac{d^{2} z}{|\zeta_{i}|^{2} \log^{2}|\zeta_{i}/a_{i}(1)|^{2}}\nonumber\\
	 &= - M_{1} + (s-1) M_{0}.
\end{eqnarray}
Comparing \eq{dilate_moments} to \eq{sum_rulesD} gives the moments for charge and dimension stated in Prop. \eq{momentsprop}. 

Carrying out this procedure for the translation sum rule \eq{translateT} implies

\begin{eqnarray}
	\mathcal{L}_{-1} &= \int_{D(p_i)} \left[ \frac{(k+2s-2)}{\log|\zeta_{i}/a_{i}(1)|^{2}} ( 2 \gamma_{i}) + O\left(\frac{1}{\log^{2}|\zeta_{i}|^{2}} \right)\right] \, \rho_{i}(|\zeta_{i}|) \frac{d^{2}z}{|\zeta_{i}|^{2} \log^{2}|\zeta_{i}/a_{i}(1)|^{2}}\nonumber\\
	 &= - 2 \gamma_{i} M_{1} + O(k^{-1})\label{translate_moments}
\end{eqnarray}
Using the value of $M_{1}$ which follows from \eq{dilate_moments}, comparing \eq{translate_moments} to \eq{sum_rulesT} implies that the correlator $\langle (\partial \varphi)(p_i) \rangle$ has the form

\begin{eqnarray}
	 \langle (\partial \varphi) (p_i) \rangle  =   \beta ( 2  \mu_{H} -  \nu a_{i})\gamma_{i}  + \sum_{j\ne i} \frac{ a_{j}}{p_i - p_j} + O(k^{-1})
\end{eqnarray}
The behavior at infinity is similar, and we present it in \ref{App_inf}.

\section{Berry Curvature on Moduli space}\label{BerryCurvatureProof}

In this section, we prove Theorem (\ref{Result1}) by computing a variational formula for the generating functional and applying the sum rules \eq{sumrulesprop}. The goal is to find the Berry curvature of the Laughlin state under adiabatic transport in the moduli space of the punctured Riemann sphere $\mathcal{M}_{0,n}$. By definition, the Berry connection one-form is $\mathcal{A}  = \mathcal{A}_{p} + \mathcal{A}_{\bar{p}}$, where
\begin{equation}
\mathcal{A}_{p} = i \sum_{i = 1}^{n-3}\langle \Psi | \partial_{p_{i}} | \Psi \rangle d p_{i}	, \quad \mathcal{A}_{\bar{p}} = i \sum_{i = 1}^{n-3} \langle \Psi | \partial_{\bar{p}_{i}} | \Psi \rangle d\bar{p}_{i}
\end{equation}
The Berry curvature is the curvature two-form of this Abelian connection $\Omega = d \mathcal{A} = \partial_{\bar{p}} \mathcal{A}_{p} + \partial_{p} \mathcal{A}_{\bar{p}}$, where $\partial_{p} = \sum_{i = 1}^{n-3} dp_i \partial_{p_{i}}$ is the Dolbeault operator. A holomorphic state is defined by its mostly holomorphic dependence on the moduli parameters (see Section \ref{holopsi})

\begin{equation}
\Psi_{L} = \frac{1}{\sqrt{\mathcal{Z}[\{p_{i}, \bar{p}_{i}\}]}} F \left( \{z_{i}, \bar{z}_{i}\}| p_{1}, ..., p_{n}\right)	
\end{equation}
which implies the Berry connection and curvature are given by

\begin{equation}\label{berrykahler}
	\mathcal{A}_{p} =  \frac{i}{2} \partial_{p} \log \mathcal{Z} , \quad \Omega = i \partial_{\bar{p}} \partial_{p} \log \mathcal{Z} = -\sum_{i,j} \frac{\partial^{2} \log \mathcal{Z}}{\partial p_i \partial \bar{p}_{j}} \, i \, dp_{i} \wedge d\bar{p}_{j} ,
\end{equation}
with the Berry phase $\gamma_{B} = \int_{R}\Omega$, where $R$ is a surface enclosed by a path in moduli space. Eq.\eq{berrykahler} shows that the generating functional \eq{Z} is a K\"ahler potential for the Berry curvature. This was observed to occur in the IQH case on a torus \cite{Levay1995}, argued to hold for the FQH case on a torus \cite{Read2009a, Read2011} and cylinder \cite{Tokatly2009}, and recently argued to be a more general property of QH states on Riemann surfaces \cite{Klevtsov2015}. In the present case, we find that this fact is tied intimately to the uniformization theorem for Riemann surfaces, and relies on the existence of the conformal maps $w(z)$ \eq{developing_map} and its inverse $J(w)$ \eq{Klein_map} that is guaranteed by this famous theorem. 

\subsection{Derivation of the Berry Curvature}
In considering the derivative of the generating functional, we are free to work in either coordinates, i.e. on $\Sigma$ or its universal cover $\mathbb{H}/\Gamma $. In order to compute the Berry connection, we stay in $\Sigma$ and find

\begin{equation}
\mathcal{A}_{p_i} =  \frac{i}{2} \frac{\partial}{\partial p_{i}} \log \mathcal{Z}	 =  \frac{i}{2} \oint_{C_{\epsilon}(p_{i})} \langle \rho \rangle e^{\phi} d\bar{z} + \frac{i}{2} \int_{\Sigma}  \partial_{p_{i}} (W + Q_{a}) \langle \rho \rangle dV .
\end{equation}
The first term must vanish since the density in a small disk around a cusp is radially symmetric with corrections of order $O(z - p_i)$. We are left to evaluate the variational formula
\begin{eqnarray*}
	\partial_{p_i} \log \mathcal{Z} 
	& = \int_{\Sigma} \partial_{p_i} W\left( \langle \rho \rangle  - \bar{\rho}\right) dV - \frac{a_{i}}{\beta} \langle (\partial_{z} \varphi)(p_i)\rangle +  \bar{\rho} \int_{\Sigma}  \partial_{p_i}W  dV .
\end{eqnarray*}
The first integral picks up the contributions coming from the cusp singularities. The last integral does not depend on the density and is only multiplied by the constant factor $\bar\rho$. We evaluate this first, and then proceed with analyzing the remaining terms. To do this, we make use of the property that the K\"ahler potential of the form $K = - \lambda \log {\rm Im} w(z) $ satisfies

\begin{equation*}
(\partial_{z} K) \partial_{\bar{z}} (\partial_{p_i} K) =  \frac{\lambda \partial_{z} w}{{\rm Im} w} \left( \lambda \partial_{p_i} w \partial_{\bar{z}}   \frac{1}{({\rm Im} w)} \right)= \partial_{p_i} K \partial_{\bar{z}} \partial_{z} K = (\partial_{p_i} K ) e^{\phi},
\end{equation*}
where we used the fact that $w(z)$ is holomorphic in $p_i$, so that $\partial_{p_i} \bar{w} = 0$. From this we find the identity

\begin{equation}
	\partial_{\bar{z}} \left( \partial_{z} K \partial_{p_i} K \right) = 2 (\partial_{p_i} K )e^{\phi}
\end{equation}
Similarly, we use

\begin{equation}
	(\partial_{p_{i}}\phi )e^{\phi} = \partial_{p_i} \left( e^{\phi}\right) = \partial_{\bar{z}} \left( \partial_{p_i} \partial_{z} K \right)
\end{equation}
The last equality follows by interchanging the order of derivatives $\partial_{\bar{z}}$ and $\partial_{p_i}$. Using these, the integral of interest becomes
\begin{eqnarray*}
	&\int_{\Sigma} (\partial_{p_i} W)e^{\phi} d^{2} z = \int_{\Sigma} \partial_{\bar{z}} \left( - \frac{B}{4} \partial_{z} K \partial_{p_i} K + (1 - s) \partial_{p_i} \partial_{z} K \right) d^{2} z\\
	& = \frac{1}{2i} \left( \oint_{|z| \to \infty} dz - \sum_{i =1}^{n-1} \oint_{C_{\epsilon}(p_i)} dz \right)  \left( - \frac{B}{4} \partial_{z} K \partial_{p_i} K + (1 - s) \partial_{p_i} \partial_{z} K \right)
\end{eqnarray*}

Taking a look at the asymptotics (\ref{kahlermod}), (\ref{kahlermodinf}), (\ref{kahlerrefined}), and (\ref{kahlerrefinedinf}), it is clear that near the punctures (and at infinity), the functions $\partial_{z} K \partial_{p_i} K$ and $\partial_{p_i} \partial_{z} K$ will be suppressed by factors $\log^{-1}|z - p_i|$ (and $\log^{-1}|z|$ at infinty). Therefore, the final result of these contour integrals will be zero, and the calculation of the Berry connection reduces to 
\begin{eqnarray}\label{var1}
	\partial_{p_i} \log \mathcal{Z} & = \int_{\Sigma}\partial_{p_i} W\left( \langle \rho \rangle  - \bar{\rho}\right) dV - \frac{a_{i}}{\beta} \langle (\partial_{z} \varphi)(p_i)\rangle.
\end{eqnarray}
At this point we employ the asymptotic expansion (\ref{rhoasymp}) to write the first integral as a sum over the local cusp densities around each puncture, giving 
\begin{eqnarray}\label{var2}
	\partial_{p_i} \log \mathcal{Z} & =\sum_{j}  \int_{D_{\epsilon}(p_j)} \partial_{p_i} W  \rho_{j} dV - \frac{a_{i}}{\beta} \langle (\partial_{z} \varphi)(p_i)\rangle
\end{eqnarray}
Next, we replace the function $\partial_{p_i}W$ by its asymptotic form near the punctures by applying the Ahlfors lemma \eq{ahlfors1} and \eq{ahlfors2} to write
\begin{eqnarray}
\partial_{p_i} W &= - \partial_{z} W \dot{F}^{i} - (1 - s) \dot{F}_{z}^{i}
\end{eqnarray}
and using the expansion of the quasiconformal map
\begin{eqnarray}
	\dot{F}^{i} &= \delta_{ij} + (z - p_j) \dot{F}_{z}^{i}(p_j) + ... \quad z \to p_j.
\end{eqnarray}
The integrals in \eq{var2} now read for $j \ne n$
\begin{eqnarray*}
\int_{D_{\epsilon}(p_j)} \partial_{p_i} W \rho_{j} dV =  \int_{D_{\epsilon}(p_j)}  \left( - \partial_{z} W \delta_{ij} - (1 - s) \dot{F}_{z}^{i}(p_j)  - (z - p_j) \partial_{z} W \dot{F}_{z}^{i}(p_j) \right) \rho_{j}dV.
\end{eqnarray*}
To get the contribution from $p_n = \infty$, we use the asymptotic formula
\begin{eqnarray}
\partial_{p_i} W = - \frac{(k + 2s - 2)}{\log|z/a_{n}(-1)|^{2}} \partial_{p_i} \log |a_{n}(-1)|^{2} + \frac{f(|z|)}{z} + ... \quad z \to p_n,
\end{eqnarray}
and evaluate the integral in the neighborhood of $p_n$, by using local coordinates $\zeta = 1/z$.  The unspecified function $f(|z|)$ will not give any contribution at order one. Employing the local moment around $p_n$ (see \ref{App_inf})
\begin{equation}
\int_{D_{\epsilon}(p_n)}  \frac{k + 2s - 2}{\log |\zeta a_{n}(-1)|^{2}} \rho_{n}(\zeta) \,\frac{\lambda d^{2}\zeta}{|\zeta|^{2} \log^{2}|\zeta a_{n}(-1)|^{2}}  = - M_{1}(p_n) ,
\end{equation}
we find that the singularity at $p_{n}$ will contribute to \eq{var1} a term 

\begin{equation}
	\int_{|z| \to p_n} \partial_{p_i} W \left( \langle \rho \rangle - \bar{\rho} \right) dV = - M_{1}(p_n) \partial_{p_i} \log |a_{n}(-1)|^{2}
\end{equation}
with $O(k^{-1})$ corrections. The Berry connection can then be written in terms of sum rules and moments as
\begin{eqnarray*}
	\partial_{p_i} \log \mathcal{Z} = - \mathcal{L}_{-1}(p_i) + &\sum_{j = 1}^{n-1} \left[ (s-1) M_{0}(p_j)  - \mathcal{L}_{0}(p_j) \right]\dot{F}_{z}^{i}(p_j)\\
	 &+ M_{1}(p_n) \partial_{p_i} \log H_{n} 
	- \frac{a_{i}}{\beta} \langle (\partial_{z} \varphi)(p_{i}) \rangle
\end{eqnarray*}
where we have used the definitions of the translation sum rule $\mathcal{L}_{-1}$ \eq{translateT} and dilatation sum rule $\mathcal{L}_{0}$ \eq{dilateT}, as well as the charge $M_{0}$ \eq{charge}, and the T-Z K\"ahler potentials \eq{Hdef}. This equation is exact up to order one with corrections appearing at $O(k^{-1})$. Evaluating using the sum rules obtained in the previous section \eq{sum_rulesT} and \eq{sum_rulesD} gives



\begin{eqnarray}
	\partial_{p_i} \log \mathcal{Z} = \frac{c_{H}}{12} \gamma_{i} - \sum_{j \ne i} \frac{\nu a_{i} a_{j}}{p_i - p_j} &+ \sum_{j} \left[ \frac{c_{H}}{24}  + h_{a_j} \right]\dot{F}_{z}^{i}(p_j)\\
	&-\left[ \frac{c_{H}}{24}  + h_{a_j} \right] \partial_{p_i} \log H_{n}
\end{eqnarray}
where we recall $h_{a} = a (2\mu_{H} - \nu a)/2$ is the dimension or spin of the flux $a$. Using then the explicit formula \eq{Fasymp} with \eq{Hdef} $\dot{F}_{z}^{i}(p_j) = \partial_{p_i} \log H_j$, and comparing to \eq{SLmodvar} we find 

\begin{eqnarray*}
	\partial_{p_i} \log \mathcal{Z} 
	= -\frac{c_{H}}{24\pi } \partial_{p_i} \mathscr{S}_{L} - \sum_{j \ne i} \frac{\nu a_{i} a_{j}}{p_i - p_j} &+ \sum_{j = 1}^{n} h_{a_{j}} \partial_{p_i} \log H_{j}
\end{eqnarray*}
where $\mathscr{S}_{L} = S_{L} - \pi \log H$ is the global K\"ahler potential on moduli space \eq{SLmod}. 
The Berry curvature can be split up in the form

\begin{equation}
	\Omega = \Omega_{g} + \Omega_{a} + O(k^{-1})
\end{equation}
where the geometric piece is presented as theorem \eq{Result1} in the introduction and is thus independent of the fluxes $a_{i}$ 

\begin{equation}\label{Omega_g}
	\Omega_{g} = - \frac{c_{H}}{24\pi} i \partial_{\bar{p}} \partial_{p} \mathscr{S}_{L} = - \frac{c_{H}}{12\pi} \left( \omega_{WP} - \frac{4\pi^{2}}{3} \omega_{TZ}\right),
\end{equation}
whereas the electromagnetic part does not depend on the central charge and is given by
\begin{eqnarray}\label{Omega_a}
\Omega_{a} 
& =   \sum_{ i<j} 2\pi \nu a_{i} a_{j} \, \delta(p_i - p_j)\, \left( \frac{i}{2} dp_i \wedge d\bar{p}_{j}\right) + \frac{8\pi}{3} \sum_{i = 1}^{n} h_{a_{i}} \omega_{TZ,i}
\end{eqnarray}

The first term in $\Omega_{a}$ can be interpreted as the mutual statistics of the quasi-holes sitting on the punctures. The conformal dimension of the quasi-holes $h_{a_i}$ was shown in \cite{Can2015} to contribute to a gravitational analog of the Aharanov-Bohm effect, producing a Berry curvature $h_{a_i} {\rm Ric}$. If the path encircles a cusp singularity $p_j$, there will be a contribution $h_{a_i} 2\pi + h_{a_{j}} 2\pi $. This interpretation is difficult to adapt directly to the present context, since the flux tubes now sit directly on top of the curvature singularities. However, it would appear that this term is a descendant of the spin-curvature coupling. It is thus natural to conjecture on physical grounds that the T-Z metric $\omega_{TZ,i}$ encodes this structure and has singularities as $p_i \to p_j$. Using this reasoning, it is possible to reproduce the formula for the total volume of the T-Z metric for $n=4$. We flesh out this argument in \ref{AppVolume}. 








\subsection{Discussion of Main result}

Here we spend some time discussing the main result, looking closer at special cases and comparing with the literature. 

\subsubsection{Braiding Cusp Singularities}

Setting the fluxes to zero, we can look at the braiding statistics of the cusp singularities. This is defined simply as the phase accumulated by the wave function under the adiabatic process of braiding one cusp around another, according to the Berry curvature (\ref{Omega_g}). 
Using the asymptotics for the accessory parameters \eq{accessory_asymp} and the T-Z potential \ref{AppVolume}, we find the Berry connection as two punctures are brought together is given by

\begin{eqnarray}
	\mathcal{A}_{p_{i}} = \frac{i}{2} \frac{c_{H}}{24} \left( 2\gamma_{i} + \partial_{p_i} \log H \right) &\to \frac{i}{2} \frac{c_{H}}{24} \left( - \frac{1}{p_{i} - p_{j}} + \frac{2}{p_i - p_j}\right),\\
	& = \frac{i}{2} \frac{c_{H}}{24} \frac{1}{p_i - p_j} , \quad p_i \to p_j,
\end{eqnarray}
Therefore, the phase accumulated by braiding cusp $p_i$ around $p_j$ is then

\begin{equation}
	B_{ij} = 2\int_{C(p_j)} dp_{i} \mathcal{A}_{p_i} = - \frac{c_{H}}{12}\pi ,
\end{equation}
where $C(p_j)$ is a contour enclosing $p_j$. In Ref. \cite{Laskin2016}, which considered piecewise flat surfaces with conical singularities, the exchange statistics was defined as half of this braiding phase with the T-Z contribution removed (since this represented the spin-curvature A-B phase). The generalization of their result for the exchange statistics for general singular geometries would be

$$\Phi_{ij} = -\pi \frac{c_{H}}{12} \frac{1}{2\pi i}\int_{C(p_{j})} \gamma_{i} dp_{i} = \pi \frac{c_{H}}{24}$$ 
For comparison, the result for conical singularities with opening angle $2\pi \alpha_{i}$ on a piece-wise flat surface is $\Phi_{ij} = \pi c_{H} \alpha_{i} \alpha_{j} /12$. Note that the result for cusps is not simply the $\alpha = 1$ limit for conical singularities on a flat surface. 
\subsubsection{Generating functional, 2D gravity, and holography}
According to \cite{Can2014,Ferrari2014}, the generating functional will have an asymptotic expansion in $k$. Using the variational formula above, combined with some recent results for the generating functional on surfaces with conical singularities \cite{Laskin2016,Klevtsov2016}, we expect this to be of the form

\begin{eqnarray}\label{gen_fun_full}
	\log \mathcal{Z} = &k^{2} A^{(2)} + k A^{(1)} + B \log k - \frac{c_{H}}{24\pi}  \mathscr{S}_{L} \nonumber\\
	&- \sum_{i <j<n} \nu a_{i} a_{j} \log |p_i - p_j|^{2} + \sum_{i=1}^{n} h_{a_{i}} \log H_{i}  + A^{(0)}+ O(k^{-1})
\end{eqnarray}
where $A^{(2)}$, $A^{(1)}$, $A^{(0)}$, and $B$ are $k$-independent constants which do not vary with the moduli $p_i$. Comparing to the large $k$ expansions in \cite{Can2014,Ferrari2014}, it appears that the true novelty of cusp singularities comes from the addition of the T-Z K\"ahler potential $\log H$, which serves to: 1) shift the Liouville action $\mathscr{S}_{L} - S_{L} = -\pi \log H$, and 2) couple to the spin or dimension of the fluxes via $h_{a_{i}} \log H_{i}$. It would be interesting to understand this better in terms of a gravitational effective action of a 2D quantum field theory.

These considerations reveal another connection which holography enthusiasts might find interesting. The modified Liouville action $\mathscr{S}_{L}$ is the renormalized volume of a particular hyperbolic 3-manifold \cite{Park2015}, and can thus be realized as the classical limit of a three dimensional gravity theory. The implication for the quantum Hall effect is not completely clear, but some thoughts in this direction were laid out in \cite{Gromov2016}, to the effect that amongst the dynamical degrees of freedom of the QHE, there is an as yet unappreciated ``gravitational" sector. In \cite{Laskin2015}, the partition function was argued to be expressed as a theory of random surfaces using collective fields. Alternatively, the generating functional was shown to follow from 2D free bosons coupled to background fields $B$ and $R$, with a string of $N$ vertex operator insertions \cite{Ferrari2014}. The presence of the Liouville action and its holographic interpretation as a 3D volume, pushed to its limit, implies that the large $N$ 2D QFT of \cite{Ferrari2014}, or the collective field theory of\cite{Laskin2015}, is dual to a 3D classical gravity. It was also suggested recently by S. Klevtsov and P. Wiegmann (see \cite{Klevtsov2016}) that in the presence of curvature singularities, the generating functional must be connected to correlation functions in quantum Liouville field theory. This interpretation is correct at the semi-classical level, since $S_{L}$ does indeed appear in $\log \mathcal{Z}$, but if there is any hope to make the connection to quantum Liouville theory, a better understanding of the T-Z potential $\log H$ is required.

\subsubsection{Conformal Dimension and Non-abelian statistics}
When all the fluxes are set to zero, $\Omega_{a} = 0$ and the Berry curvature can be expressed in terms of the dimension of the singularities

\begin{equation}\label{Berry_general}
\Omega = - \frac{2 M_{1}	}{\pi} \left( \omega_{WP} - \frac{4\pi^{2}}{3} \omega_{TZ}\right).
\end{equation}
 The details of the particular FQH wave function used in the calculation will affect the coefficient $M_{1}$, but not the symplectic form which it multiplies. The dimension, while an integrated quantity, is essentially a property of a localized impurity, the cusp singularity. This most likely means that our result should only hold for physical states that exhibit the main features of the screening behavior of Laughlin states, such as the exponential suppression of density correlation functions in the bulk.

Indeed, Eq. \eq{Berry_general} is a possible way to generalize our main result for Laughlin states to general holomorphic LLL states. According to (\ref{cohom}), the dimension $M_{1}$ is a topological invariant. However, it is not clear that it is universal, i.e. invariant on a plateau. If the remainder $O(k^{-1})$ terms in the geometric curvature are indeed exact forms as argued in the case of smooth geometry in \cite{Klevtsov2015}, then the only way to change this topological invariant would be through the local moments of the density. 

Perhaps a hint at the structure of the Berry curvature comes when we include flux tubes. Then, the coefficient of $\omega_{TZ}$ is modified, but $\omega_{WP}$ is unaffected. Another way to think of flux tubes is that they make the punctures distinguishable. For instance, the state with $a_{1}$ at $p_{1}$ and $a_{2}$ and $p_{2}$ is distinct and orthogonal to $a_{1}$ at $p_{2}$ and $a_{2}$ and $p_{1}$. Consequently, for a given vector $(a_{1}, ..., a_{n})$, and $\sum_{i} a_{i}>0$, there will be a degenerate subspace of states. Adiabatic transport of punctures will then cycle between these states, giving rise to non-abelian statistics (c.f. the case of genons in \cite{Gromov2016}). The universal contribution to the abelian Berry phase which is independent of the fluxes is the integrated W-P metric.  

Let us rephrase this observation in somewhat more physical terms. If the magnetic field is tuned slightly away from commensurability, there will be a proliferation of quasi-holes. One may reasonably assume that the singularities act as a sort of impurity potential and trap the quasi-holes. This will then lead to a degenerate ground state subspace mentioned above, and opens up the possibility of non-abelian braiding of punctures.

\subsubsection{Large spin}\label{LargeSpin}

In Ref.\cite{Levay1999b}, the vector potential $A_{\bar{z}} = (i k /4) \partial_{\bar{z}} \phi$ with $s = 0$ was used in \eq{LLLconstraint}, so that the LLL wave functions were zero modes of the operator 

\begin{equation}
	D_{\bar{z}} = \partial_{\bar{z}}  - i \frac{e}{\hbar}  \left( \frac{i k}{2} \partial_{\bar{z}} \phi \right) = \partial_{\bar{z}} + \frac{k}{4} \partial_{\bar{z}} \phi
\end{equation}
According to our formalism, this would correspond to moving all the magnetic flux into the spin by setting $s = k/2$. However, the curvature singularities will lead to a very large flux concentrated at the punctures which we can offset by setting $a_{i} = - s$. The parameters in the problem become

\begin{eqnarray}
	c_{H} = 1 - 3 \nu^{-1} + 6 k - 3 \nu k^{2}, \quad \mu_{H}  = \frac{1}{2}  - \nu \frac{k}{2}, \quad h_{a} =  \frac{k}{4}\left( \nu \frac{k}{2} - 1\right)\\ 
	\frac{c_{H}}{24}+ h_{a} = \frac{1 - 3 \nu^{-1}}{24}
\end{eqnarray}

The quasi-hole statistics part of $\Omega_{a}$ would no longer show up in the adiabatic connection, and we would be left with the generating functional



\begin{equation}
	\log \mathcal{Z} = \frac{(3 \nu k^{2} - 6 k)}{24\pi} S_{L} - \frac{(1 - 3 \nu^{-1})}{24\pi} \mathscr{S}_{L},
\end{equation}
with the Berry curvature 

\begin{equation}
	\Omega = \frac{3 \nu k^{2} - 6 k}{12\pi} \omega_{WP} - \frac{(1 - 3 \nu^{-1})}{12\pi} \left( \omega_{WP} - \frac{4\pi^{2}}{3} \omega_{TZ}\right).
\end{equation}
Specifying for $\nu = 1$, $k = B \lambda/2$ and $\lambda = 4$ for curvature $R = -2$, we find

\begin{equation}
	\Omega = -2 \left[ -\frac{6 B^{2} - 6 B + 1}{12\pi} \omega_{WP} + \frac{\pi}{9} \omega_{TZ}\right].
\end{equation}

The expression in brackets was the result for the adiabatic curvature quoted in \cite{Levay1999b}. We suspect the factor of $-2$ reflects a slight difference in conventions. 

Of course, the problem we solve is slightly different from computing the variation of the determinant of $L^{2}$ norms of holomorphic sections, which is essentially what is computed in \cite{Levay1999b}. The main reason is that we have concentrated on small spin states, and thus have not allowed for sections that have poles at the punctures. The discrepancy can be seen simply by looking at the degeneracy of the LLL, which for large spin must read $N = (1 - 2s) +n (s - 1) $.
However, by merely substituting $s \to k/2$ and $a \to  -s$ and $k \to 0$, we have	$N = (1 - 2s) - n a = (1 - 2s) + n s$ which has $n$ too many states. Thus it appears that our problem can be mapped to a relative of Ref.\cite{Levay1999b} which has fluxes $a = 1$ piercing each puncture.

\subsubsection{Gravitational Anomaly}

The gravitational anomaly or chiral central charge, given by $c = c_{H} + 12 \nu^{-1} \mu_{H}^{2}$ can be isolated by adjusting the fluxes to cancel the total charge $M_{0} = \mu_{H} - \nu a = 0$ at each cusp. This requires setting $a = \nu^{-1} \mu_{H}$, and gives for the Berry curvature


\begin{equation}
	\Omega_{g} = - \frac{c_{H}}{12\pi} \omega_{WP}    +c\frac{ \pi }{9} \omega_{TZ} - \sum_{i<j} 2\pi \nu^{-1} \mu_{H}^{2} \delta(p_i - p_j) \left( \frac{i}{2} dp_i \wedge d\bar{p}_{j}\right)
\end{equation}
Thus, the chiral central charge remains the only coefficient controlling the appearance of the T-Z metric.

A more obvious possibility is to take all the fluxes to zero, and adjust $\mu_{H} =0$ by setting the spin to $s = \nu^{-1}/2$. This clearly leads to

\begin{equation}
\Omega = - \frac{c}{12\pi} \left( \omega_{WP} - \frac{4\pi^{2}}{3} \omega_{TZ}\right),
\end{equation}
in which case the chiral central charge is the whole effect. 

\section{Exact Results for Free Fermions}\label{IQHexact}

Here we compare some of our results against exact results obtained from the IQH states on a punctured disk. In this setting, the mean density is proportional to the Bergman kernel on the diagonal, which was recently studied on a punctured Riemann surface in the recent papers \cite{Auvray2016, Sun2016}. We rephrase some of their results in the language of physics, including spin $s$ and flux $a$, and with an emphasis on the moments. 

On a punctured unit disk, $\Sigma  = \mathbb{D}/\{ 0\}$, the constant negative curvature metric is the Poincar\'e metric

\begin{equation}
	dV = \frac{\lambda |dz|^{2}}{|z|^{2} \log^{2}|z|^{2}} 
\end{equation}
with scalar curvature $R = - 8/\lambda$ on $\Sigma$ and infinite volume. Recall the notation $k = -4B/R =  B\lambda / 2$. The normalized single-electron eigenstates are

\begin{equation}
|\psi_{n}(z)|^{2} = \frac{(n+a + j)^{k+2s - 1}}{\pi \lambda  \Gamma(k + 2s - 1)}  |z|^{2(n + s + a)} \left|\log |z|^{2}\right|^{k + 2s} 
\end{equation}
with $n + s + a >0$, and $n \in \mathbb{Z}$. There is a bit of a subtle point concerning the range of $n$. If we let $a = s = 0$, it turns out that $\psi_{0}$ is not normalizable. The problem here is that $\psi_{0}(0) = \infty$, so the wave function is not regular at the singularity. We are left with two options: regulate this divergence or throw out the $n = 0$ state. We choose to keep $n = 0$ but ensure the vanishing of the wave function at the origin by including either positive spin $s$, positive flux $a$, or positive combination $a+s$. When $a+s$ is large enough to allow for $n<0$, we can simply redefine $n' = n + \left\lfloor{a+s}\right\rfloor$ where $\lfloor{a+s}\rfloor$ is the integer part of $a+s$. This means that the mean density

\begin{equation}
	\langle \rho(z) \rangle = \sum_{n = n_{min}}^{\infty} |\psi_{n}(z)|^{2}
\end{equation}
is periodic in $a$, in the sense that $\rho_{a+ m} = \rho_{a}$ for integer $m$. We will therefore simply use $n_{min} = 0$ with positive $a+s$. The mean density can then be written succinctly as

\begin{equation}
	\langle \rho(x) \rangle = \frac{x^{a+s} | \log x|^{k+2s}}{\pi \lambda \Gamma(k+2s-1)} \Phi(x, 1 - k - 2s, a+s)
\end{equation}
where $x = |z|^{2}$ and $$\Phi(x,s,p) = \sum_{n= 0}^{\infty} \frac{x^{n}}{ (n +p )^{s}}$$ is the Lerch transcendent. The moments now involve integrals over this special function.

The density achieves the constant value away from the cusp

\begin{equation}
\rho \to \bar\rho =  \frac{k - 1 + 2s}{ \pi \lambda }  = \frac{B}{2\pi}   + \frac{(1 - 2s)}{8\pi } R	
\end{equation}

The charge moment can be computed without the $\epsilon \to 0$ limit in the definition \eq{charge}.

\begin{prop}
	
The charge at the cusp singularity in the IQH state is

\begin{equation}\label{Qexact}
	M_{0} = \int_{\mathbb{D}} \left( \langle \rho(z)\rangle  - \bar{\rho}\right) dV = - a + \frac{1}{2}(  1- 2s)
\end{equation}

\end{prop}

\begin{proof}
 To prove the charge moment requires a recursion relation relating $M_{0}$ at different $k$. For shorthand, we write $p = k+2s$, and $q = a+s$. Under a change of variable $y = - \log x$ (which maps the disk to the upper half plane), the integral \eq{Qexact} becomes 
\begin{equation}
	M_{0}(p) = (p - 1) \int_{0}^{\infty} \frac{dy}{y^{2}} \left( \frac{y^{p}}{\Gamma(p)} e^{ - q y} \Phi(e^{ - y}, 1 - p, q) - 1\right).
\end{equation}
Next we utilize the identity

\begin{equation}\label{lerchidentity}
	\partial_{y} \left(e^{ - q y} \Phi(e^{ - y}, 1 - p, q)\right) = - e^{ - q y} \Phi(e^{ - y}, 1 - (p+1), q),
\end{equation}
to relate the density, and consequently the moment $M_{0}(p)$, at $p$ to that at $p+1$. An integration by parts will allow us to hit this special function with a derivative, resulting in the simple relation

\begin{equation}
	M_{0}(p) =  M_{0}(p+1)
\end{equation}
Remarkably, the sum rule is valid for all (integer) $p$, so we may solve it for the simplest case at $p= 2$ for which the density $\rho_{p}$ is

\begin{equation}
	\rho_{2}(y) = \frac{ y^{2}}{\pi \lambda } \sum_{n} (n+q)e^{-(n+q)y} = -\frac{y^{2}}{\pi \lambda} \partial_{y}\left( \frac{e^{ - p y}}{1 - e^{-y}}\right).
\end{equation}


Computing the charge for $p = 2$ then gives
\begin{eqnarray}
	M_{0}(2) 
	& = \lim_{\epsilon \to 0}\int_{\epsilon}^{\infty} \frac{dy}{y^{2}} \left( -y^{2} \partial_{y} \left( \frac{e^{ - q y}}{1 - e^{-y}}\right) - 1\right)\\
	& =\lim_{\epsilon \to 0} \left[ \frac{e^{ - q \epsilon}}{1 - e^{ - \epsilon}} - \frac{1}{\epsilon}\right] = \frac{1}{2} - q.
\end{eqnarray}

Therefore $M_{0}(p) = M_{0}(2) = \frac{1}{2} - q$.

\end{proof}
 \begin{figure}[htbp]
   \includegraphics[scale=1.]{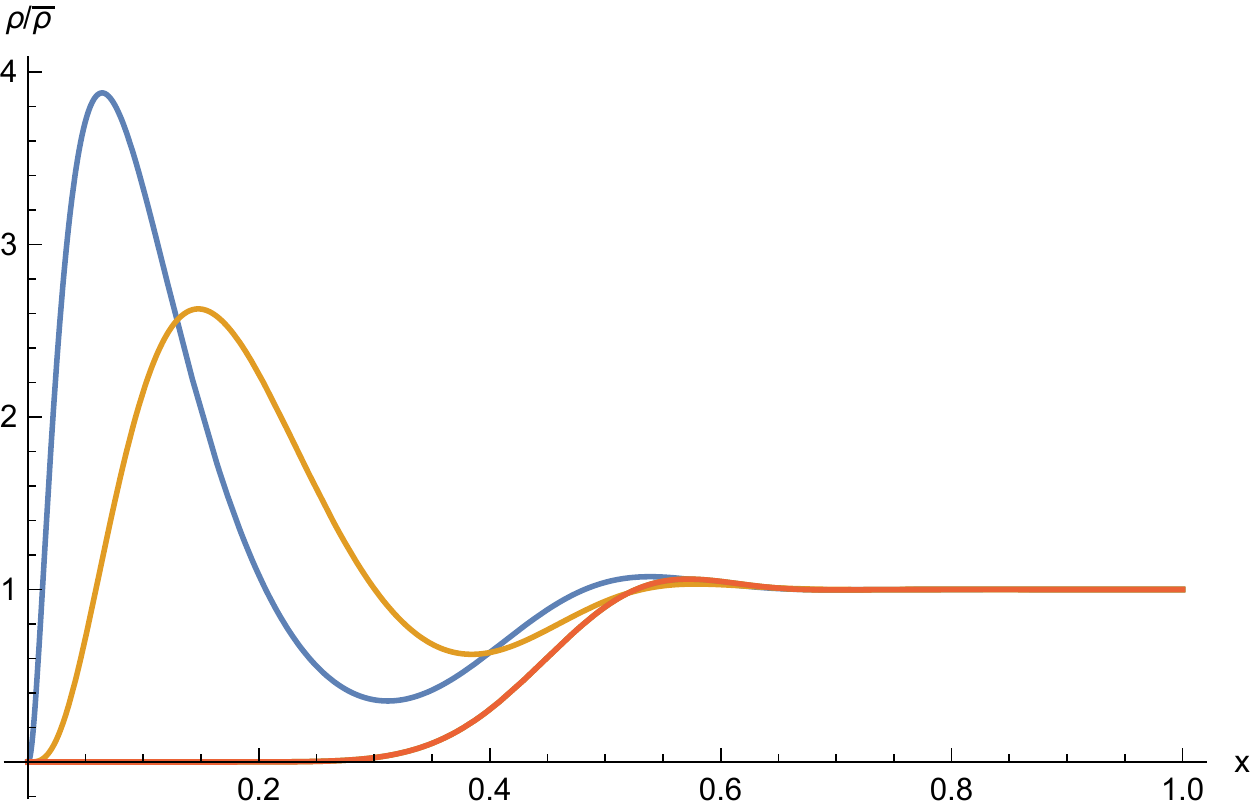}
 %
   \caption{  Normalized density in rescaled coordinates $\rho(x^{k+2s+1})$ \cite{Auvray2016} for $s = 1/3$, $M_{0} = +1/6$ (blue), $s = 1/2$, $M_{0} = 0$ (yellow) and $s = 3/2$, $M_{0} = -1$ (red). }
      \label{Fig2}
   \end{figure}
The dilatation moment for the dimension can be computed as well. We present it as another proposition. 

\begin{prop}
	
The dimension of the cusp singularity in the IQH state is

\begin{equation}\label{Lp}
	M_{1} =  -\int_{\mathbb{D}} \frac{(k+2s-2)}{\log |z|^{2}}  \left( \langle \rho(z)\rangle  - \bar{\rho}\right) dV_{z} = \frac{c_{H}}{24} + h_{a}
\end{equation}
where $c_{H} = - 2 + 6 s - 6 s^{2}$ and $h_{a} = (a/2)(1 - 2s - a)$. 
\end{prop}

\begin{proof}

We proceed as with the charge moment to obtain a recursion relation for $M_{1}(p)$. Using the same coordinate transformation as above, the integral for the dimension is written

\begin{equation}
	M_{1}(p) =  (p - 2) (p-1)  \int_{0}^{\infty} \frac{dy}{y^{3}} \left( \frac{y^{p}}{\Gamma(p)} e^{ - q y} \Phi(e^{ - y}, 1 - p, q) - 1\right)
\end{equation}
An integration by parts utilizing \eq{lerchidentity} gives the relation

\begin{equation}\label{Lrec}
	M_{1}(p) =  M_{1}(p+1) 
\end{equation}
So that once again $M_{1}(p)$ as defined is independent of $p$. We choose for the base case $M_{1}(3)$, and use

\begin{equation}
	e^{ - q y} \Phi(e^{ - y}, -2, q) = \sum_{n = 0}^{\infty} (n + q)^{2} e^{ - (q + n) y}   = \partial_{y}^{2} \left( \frac{e^{ - q y}}{1 - e^{ - y}}\right)
\end{equation}
Then the integral for $M_{1}(3)$ is easily solved to give

\begin{equation}
M_{1}(3) = -\lim_{\epsilon \to 0} \left[ \partial_{y} \left( \frac{e^{-q y}}{1 - e^{ - y}}\right) + \frac{1}{ \epsilon^{2}} \right]	 = -\frac{1}{12}(1 - 6 q + 6 q^{2})
\end{equation}
This can be unpackaged using $q = a + s$ to give the stated proposition \eq{Lp}.

\end{proof}

In this same manner, we can find a generating function for the $n^{th}$ moment, defined to be

\begin{equation}
	M_{n} =  \frac{\Gamma(k+2s - 1)}{\Gamma(k+2s - n - 1)}  \int_{\mathbb{D}} \frac{1}{( - \log |z|^{2})^{n}} \left( \langle \rho(z)\rangle - \bar{\rho}\right) dV_{z}
\end{equation}
This normalization ensures that $M_{n}$ is independent of $k$. The generating function for the moments is given by 

\begin{equation}
F_{q}(y) = \frac{e^ { - qy }}{1 - e^{ - y}} - \frac{1}{y}	, 
\end{equation}
and let us denote its derivatives $\partial_{y}^{n}F_{q}(y) = F_{q}^{(n)}(y)$. The expression for the moments comes from the coefficients of the series expansion of $F_{q}$ around $y = 0$. Specifically, we find the simple formula 
\begin{equation}
	M_{n} = (-1)^{n} F_{q}^{(n)}(0). 
\end{equation}



\section{Conclusion}

Our main result is the Berry curvature on the moduli space of punctured Riemann spheres, which is finite in the limit of large $N$ and is controlled by the gravitational anomaly. Our results are strictly proved for Laughlin wave functions, but the question of how they can be generalized is equivalent to the question of what values the geometric response coefficients $\mu_{H}$ and $c_{H}$ take for different FQH states.  
It remains to understand the universality of the topological invariant $c_{H}$ constructed from this Berry curvature. We have related this coefficient to an integrated moment of mean density, which would provide another route to probe the central charge and ascertain universality. 


Our results give a novel physical interpretation of the Takhtajan-Zograf metric, and call for a more detailed understanding of the T-Z K\"ahler form. In particular, a detailed understanding of these potentials would shed light on the connection of the FQH state to Liouville quantum gravity. 



\appendix

\section{}\label{App_asymp}
\section*{Refined Asymptotics Near Punctures}

In this appendix, we list some refined asymptotics which are necessary for some of the arguments made in Secs. (\ref{BerryCurvatureProof}) and (\ref{WardIdentity}). All of the asymptotic formula for the metric and local K\"ahler potential follow from the plugging the expansion (\ref{dev_map}) into the expressions \eq{uniform_metric} and \eq{uniform_kahler}. We have near the punctures, 

\begin{eqnarray}
	\phi =& - \log |\zeta_i|^{2} - 2 \log \left | \log | \zeta_i / a_{i}(1)|^{2} \right| \nonumber \\
	 &+ 2 {\rm Re}(\gamma_{i} \zeta_{i}) \left( 1 - \frac{1}{\log |\zeta_{i}/a_{i}|}\right)  + O(\zeta_{i}^{2}), \quad z \to p_i\label{cusprefined}\\
	K =& - \lambda \log \left | \log | \zeta_i / a_{i}(1)|^{2} \right|  - \frac{ \lambda  {\rm Re}(\gamma_{i} \zeta_{i})}{\log |\zeta_{i}/a_{i}|}  + O(\zeta_{i}^{2} \log^{2}|\zeta_i|), \quad z \to p_i, \label{kahlerrefined}
\end{eqnarray}
and at infinity, 
\begin{eqnarray}\label{cusprefinedinf}
	\phi   =&- \log |z|^{2} - 2 \log \log |z/a_{n}(-1)|^{2} \nonumber\\
	&+ 2 {\rm Re}(\gamma_{n}/z) \left( 1 +  \frac{1}{\log |z/a_{n}(-1)|}\right) + O(z^{-2}), \quad z \to p_n,\\
	K   =&- \frac{V}{\pi}\log |z|^{2} - \lambda  \log \log |z/a_{n}(-1)|^{2} \nonumber\\
	&+  \frac{\lambda  {\rm Re}(\gamma_{n}/z)}{\log |z/a_{n}(-1)|}
	 + O(z^{-2}), \quad z \to p_n.\label{kahlerrefinedinf}
\end{eqnarray}

For the first derivative near the punctures, we have
\begin{eqnarray}\label{dcusprefined}
	\partial_{z} \phi &= - \frac{1}{\zeta_{i}} + \gamma_{i} - \frac{1}{\zeta_{i} \log |\zeta_{i}/a_{i}(1)|} \\
	&- \frac{1}{\zeta_{i} \log |\zeta_{i}/a_{i}(1)|}\left( \gamma_{i} \zeta_{i}  - \frac{{\rm Re}(\gamma_{i} \zeta_i )}{\log |\zeta_i / a_i (1)|} \right) + O(\zeta_{i}) , \, \, z \to p_i  \\
	\partial_{z} K & = - \frac{\lambda }{ 2\zeta_{i} \log |\zeta/a_{i}(1)|} \left(1 + \gamma_{i} \zeta_{i}  - \frac{{\rm Re}(\gamma_{i} \zeta_i )}{\log |\zeta_i / a_i (1)|} + O(\zeta_{i}^{2}) \right) \,\, z \to p_i ,\label{dcusprefined2}
\end{eqnarray}
while the asymptote at infinity is
\begin{eqnarray}\label{dcuspinf}
	\partial_{z} \phi &= - \frac{1}{z} - \frac{1}{z \log |z/a_{n}(-1)|} - \frac{\gamma_{n}}{z^{2}}  + O(z^{-3}), \quad z \to p_n\\
	\partial_{z} K & =  \frac{\lambda (n-2)}{2 z} - \frac{\lambda }{2z \log |z/a_{n}(-1)|} + O(z^{-2} \log^{-2}|z|)\label{dkahlerinf}
\end{eqnarray}
Since these asymptotics appear inside integrals, it is important not only to keep track of the order of vanishing (given by the magnitude), but also the order of the harmonic (given by the argument of the complex number). 

\section{}\label{App_inf}
\section*{Puncture at infinity}

The particular normalization we choose for the punctures $X = \{p_1, ..., p_{n-3},0,1,\infty\}$ produces some technical complications around the point at infinity $p_n$. The idea here is that $p_n$ is not in any way different from the other points, but our choice of coordinates and divisor $X$ make it distinguished. This requires treating infinity in many of the formulas separately. Much of the technical details in this appendix do not pertain directly to our main result, but we feel it would be negligent to leave them out of the paper entirely.

This appendix is dedicated to verifying the fact that the puncture at infinity is not an exceptional point. To accomplish this, we will compute the local moments \eq{moments} around $p_n$ and show that they conform to the formula given by Prop. \ref{momentsprop}.

\paragraph{Generating function at infinity}

We begin by listing the asymptotic behavior of the generating function as $z \to p_n$. As one might guess, this will be used later on to prove the local moments around $p_n$. This requires using the asymptotes \eq{dcuspinf} and \eq{dkahlerinf} in the formula for the generating function (\ref{T1line1})-(\ref{T1line3}), to find up to order $O(1)$

\begin{eqnarray}\label{T1inf}
\mathcal{T}(z) = \frac{\mu_{H} B}{4} \left( - \frac{\lambda (n-2)^{2}}{2} \frac{1}{z^{2}} -\frac{\lambda (n-2) \gamma_{n}}{z^{3}} + ...\right) \\
 + (s - 1) \mu_{H} \left( \frac{n-1}{z^{2}} + \frac{2\gamma_{n}}{z^{3}} \right)  - \frac{c_{H}}{24} \left( \frac{1}{ z^{2}} + \frac{2\gamma_{n}}{z^{3}} + ... \right)	\\
+\frac{1}{2\beta} \frac{|{\bf a}|( |{\bf a}| + 2 - \beta)}{z^{2}} + \frac{2 ({\bf p} \cdot {\bf a}) \left( |{\bf a}| + 2 - \beta\right)}{2 \beta z^{3}}\\
   - \sum_{i}\nu a_{i}\left( \frac{1}{z} + \frac{p_i}{z^{2}} + \frac{p_{i}^{2}}{z^{3}} + ...\right) \langle \partial \varphi(p_i) \rangle + O(z^{-4}) \label{T1inf2}
\end{eqnarray}

\paragraph{Moments around Infinity}
Taking $z \to \infty$, the generating function behaves

\begin{eqnarray*}
	\mathcal{T}(z) &= \frac{1}{z} \int_{|\xi|<|z|}\left( 1+ \frac{\xi}{z} + \frac{\xi^{2}}{z^{2}} + ... \right) \partial W \left( \langle \rho \rangle - \bar{\rho}\right) dV \\
	& = \frac{1}{z} \int_{\Sigma }\left( 1+ \frac{\xi}{z} + \frac{\xi^{2}}{z^{2}} + ... \right) \partial W \left( \langle \rho \rangle - \bar{\rho}\right) dV   \\
	& -\frac{1}{z} \int_{|\xi|>|z|}\left( 1+ \frac{\xi}{z} + \frac{\xi^{2}}{z^{2}} + ... \right) \partial W \left( \langle \rho \rangle - \bar{\rho}\right) dV 
\end{eqnarray*}
In the second line, the domain of integration is the entire punctured sphere. These will involve global sum rules (as opposed to the local sum rules (\ref{sumrules}). Here we focus on the global dilatation sum rule, since we have at hand an exact formula. The equation which we wish to solve for the local moments around infinity will then

\begin{equation}\label{sumruleinf}
	\frac{1}{2\pi i} \oint\,  z \mathcal{T}(z) dz = \int_{\Sigma} \xi \partial_{\xi} W \left( \langle \rho \rangle - \bar{\rho}\right) dV - \int_{|\xi|> |z|} \xi \partial_{\xi} W \left( \langle \rho \rangle  - \bar{\rho}\right) dV
\end{equation}

The last integral contains the essential information about the local moments. To see this, we must utilize the large $z$ asymptotics in the integrand

\begin{eqnarray*}
	\partial_{z} W &= - \frac{B}{2} \partial_{z} K + (1 - s) \partial_{z} \phi\\
	& = \left(k + 2s - 2\right) \frac{1}{z \log |z/a_{n}(-1)|^{2}}\left( 1 + \frac{\gamma_{n}}{z} + \frac{{\rm Re}(\gamma_{n}/z)}{\log |z/a_{n}(-1)|} + ... \right)\\
	& - \frac{N_{\phi}}{ z} + (1 - s) \left( - \frac{1}{z} - \frac{\gamma_{n}}{z^{2}} + ...\right), \quad z \to p_n
\end{eqnarray*}

As well as for the metric

\begin{eqnarray*}
	e^{\phi} &= \frac{\lambda }{|z|^{2} \log^{2}|z/a_{n}(-1)|^{2}} \left( 1 + 2 {\rm Re}(\gamma_{n}/z) \left( 1 + \frac{1}{\log|z/a_{n}(-1)|} \right)+...\right]\\
	\partial_{z}\phi &= - \frac{1}{z} - \frac{\gamma_{n}}{z^{2}} - \frac{1}{z \log |z/a_{n}(-1)|}\left( 1 + \frac{\gamma_{n}}{z} + \frac{{\rm Re}(\gamma_{n}/z)}{\log |z/a_{n}(-1)|} + ... \right)\\
	\partial_{z} K & = \frac{\lambda (n-2)}{2 z} -\frac{\lambda}{2} \frac{1}{z \log |z/a_{n}(-1)|}\left( 1 + \frac{\gamma_{n}}{z} + \frac{{\rm Re}(\gamma_{n}/z)}{\log |z/a_{n}(-1)|} + ... \right)
\end{eqnarray*}

Then in the integrand, we transform to local coordinates $\zeta = 1/\xi$, in which the density is expressed using the asymptotic cusp density $\rho_{n}(\zeta)  = \langle \rho(1/\zeta) \rangle - \bar{\rho} $, and we find the sum rule related to the moments via
\begin{eqnarray}
	-\int_{|\xi|> z} \xi \partial_{\xi} W \left( \langle \rho \rangle - \bar{\rho} \right) dV_{\xi} &= \lim_{z \to \infty} \int_{|\zeta| <1/|z|}  \left(  \frac{(k + 2s - 2)}{\log|\zeta a_{n}(-1)|^{2}}  + N_{\phi} +1- s \right) \rho_{n} dV_{\zeta}\nonumber\\ 
& = - M_{1}(p_n) +  \left( N_{\phi} + 1 - s\right)M_{0}(p_n) + O(k^{-1}). \label{sumrule_momentsinf}
\end{eqnarray}

As in Sec. (\ref{mom_from_sum}), the angular integration will kill many terms in the integral. This means Eq. (\ref{sumruleinf}) implies

\begin{equation}
M_{1}(p_n) - \left( N_{\phi} + 1 - s\right)M_{0}(p_n) = \frac{1}{2\pi i} \oint\,  z \mathcal{T}(z) dz - \int_{\Sigma} \xi \partial_{\xi} W \left( \langle \rho \rangle - \bar{\rho}\right) dV
\end{equation}

Now we proceed to compute the RHS. The first term is the generating function and comes from the Ward identity. From \eq{T1inf}-\eq{T1inf2}, it will be
\begin{eqnarray}
	\frac{1}{2\pi i} \oint z \mathcal{T} dz =& - \frac{\mu_{H} B \lambda (n-2)^{2}}{8}- \frac{c_{H}}{24}    + \frac{\nu}{2}|{\bf a}| ( |{\bf a}| - \beta  + 2s )\nonumber \\
	&+ (s - 1) \left( \mu_{H} (n-1) - \nu | {\bf a}|\right)- \sum_{i} \nu a_i p_i \partial \varphi(p_i)\label{zTinf}
\end{eqnarray}
where $| {\bf a}| \equiv \sum_{i = 1}^{n-1} a_{i}$. The second term follows from an exact global dilatation sum rule which reads

\begin{equation}
	0  = \frac{N}{2} \left( 2 + \beta N - \beta + 2 |{\bf a}|\right) + \int z \partial_{z} W \langle \rho \rangle dV_{z} + \sum_{i = 1}^{n-1} \nu a_i p_i \langle \partial \varphi(p_i) \rangle
\end{equation}
Subtracting the contribution from the mean bulk density $\bar{\rho}$

\begin{eqnarray}
	\bar{\rho} \int z \partial_{z} W dV &= - \bar{\rho} V \left( \frac{1}{2} N_{\phi} + 1 - s\right) \\
	&= - \left( \nu N_{\phi} + \mu_{H} (2-n )\right) \left( \frac{1}{2} N_{\phi} + 1 - s\right)
\end{eqnarray}
We get
\begin{eqnarray}
		\int_{\Sigma} z \partial_{z} W \left( \langle \rho \rangle - \bar{\rho}\right) dV_{z}  & = { - \frac{\mu_{H} N_{\phi}(n-2)}{2}} -  N_{\phi} \left(  \mu_{H} -  \nu a_{n}\right)\\
		& { + (s-1) \left( (n-1) \mu_{H}  -\nu |{\bf a}|  \right)} \\
	& +(s-1) \left(  \mu_{H}  - \nu a_{n} \right) + \frac{\nu}{2} ({ |{\bf a}|^{2}} - a_{n}^{2})\\
	& - \mu_{H}  ({ |{\bf a}|} - a_{n})\\
		& - \sum_{i = 1}^{n-1} \nu a_i p_i \langle \partial \varphi(p_i) \rangle
\end{eqnarray}
Combining this with \eq{zTinf} and using \eq{sumruleinf}  and \eq{sumrule_momentsinf} implies

\begin{equation}
M_{0}(p_n) = \mu_{H} - \nu a_{n}, \quad M_{1}(p_n) = \frac{c_{H}}{24} + \frac{a_{n}}{2} \left( 2\mu_{H} - \nu a_{n}\right),
\end{equation}
which is what we wanted to show.

\section{}\label{AppVolume}
\section*{Volumes of T-Z and W-P metrics}

In \cite{Can2015}, we found that the Berry phase of a flux with dimension $h_{a}$ moving in the presence of curvature will receive a gravitational A-B phase contribution 

\begin{equation}
	\gamma_{g-AB} = -h_{a} \frac{1}{2}\int_{C} R dV
\end{equation}

Encircling a cusp singularity would then imply a phase $\gamma_{g-AB} = -2\pi  h_{a}$. On the other hand, if we are braiding flux-curvature composites, then braiding $p_{i}$ around $p_{j}$ will pick up contributions from both ``spins" interacting with the other cusp, therefore $\gamma_{g-AB} = - 2\pi h_{a_{i}}  - 2\pi h_{a_{j}}$. 

Comparing this to the Berry curvature \eq{Omega_a}, we identify {\it by analogy}

\begin{equation}
	\frac{8\pi}{3} h_{a_{i}} \omega_{TZ, i}  \leftrightarrow - h_{a_{i}} \frac{1}{2} R dV
\end{equation}
This suggests $\int \omega_{TZ,i}  =-3/4 $ if the path of $p_i$ on $\mathcal{M}_{0,n}$ encloses $p_j$. In this appendix, we take this suggestion seriously and show that it implies the volume formula for $n=4$ \eq{volumeTZ} for the T-Z metric.

Let $p_{1} = p$, and take the spin $h_{a_1}$ around another puncture $p_{j}$. Our hypothesis is that the integrated T-Z metric over a region $R_{j}$ in the moduli space $\mathcal{M}_{0,4}$ enclosing $p_{j}$ is
\begin{eqnarray*}
	\int_{R_{j}} \omega_{TZ, 1} = -\frac{3}{4\pi} \int_{R_{j}} d^{2}p  \frac{\partial^{2} \log H_{1}}{\partial p \partial \bar{p}} =  -\frac{3}{4\pi} \frac{1}{2i} \oint_{\partial R_{j}} dp \partial_{p} \log H_{1} = -\frac{3}{4}\\
	\int_{R_{n}} \omega_{TZ, 1} = -\frac{3}{4\pi} \frac{1}{2i} \oint_{|p| \to \infty}  dp \partial_{p} \log H_{1}  =- \frac{3}{4}
\end{eqnarray*}
where $R_{n}$ is a region concentrated around the point at infinity. These integrals imply that the leading asymptotic behavior of the T-Z potential is
\begin{eqnarray}
	\partial_{p} \log H_{1}(p, p_{2}, p_{3}) &=  \frac{1}{p - p_j}, \quad p \to p_{j}\\
	& =  \frac{1}{p} , \quad p \to p_n
\end{eqnarray}
The Berry phase also picks up a contribution from spin $h_{a_j}$ moving around the curvature singularity at  $p_{1}$. This implies $\int_{C_{j}} \omega_{TZ, j} = -3/4$, and by similar reasoning above, the leading asymptotes 

\begin{eqnarray}
 \partial_{p} \log H_{j}(p, p_2, p_3) =  \frac{1}{p - p_{j}}, \quad p \to p_j , \quad j \ne 1, n\\
 \partial_{p} \log H_{n} =  -\frac{1}{p}, \quad p \to p_n.
\end{eqnarray}
Furthermore, these should be the only singularities of the T-Z potentials. Using these asymptotes, we may then compute the total volume of each cusp metric
\begin{eqnarray*}
	\int_{\mathcal{M}_{0,4}} \omega_{TZ, 1} =- \frac{3}{4\pi} \frac{1}{2i} \left[  \oint_{|p| \to \infty} dp \partial_{p} \log H_{1} - \sum_{j = 1}^{n-2} \oint_{\partial R_{j}} dp \partial_{p} \log H_{1} \right] = \frac{3}{4} (n-3 )\\
	\int_{\mathcal{M}_{0,4}} \omega_{TZ, j} = -\frac{3}{4\pi} ( - \frac{1}{2i}) \oint_{R_{j}}  dp \partial_{p} \log H_{j} = \frac{3}{4}, \quad j \ne 1\\
	\int_{M_{0,4}} \omega_{TZ, n} = - \frac{3}{4\pi} \frac{1}{2i} \oint_{R_{n}} dp \, \partial_{p} \log H_{n} = \frac{3}{4}
\end{eqnarray*}
where the last contour integral is taken clockwise over the $C_{n}$ whose radius is sent to infinity. Putting this together, we get

\begin{equation}
	\int_{\mathcal{M}_{0,4}}  \omega_{TZ} = \int_{\mathcal{M}_{0,4}} \sum_{j = 1}^{4} \omega_{TZ,j} = \frac{3}{4} ( 2n-4 ) = 3,
\end{equation}
which is in fact the expected result \eq{volumeTZ}. Remarkably, using the analogy to the gravitational Aharanov-Bohm effect, we were able to derive the total volume of the T-Z metric.

For completeness, we also derive the volume of the Weil-Petersson metric from the asymptotics of the accessory parameters, repeating the arguments of \cite{Zograf1990a} for the case of $n = 4$. The accessory parameters have the leading behavior as punctures merge 
\begin{eqnarray*}
	\gamma_{1}(p,p_{2}, p_{3}) &= - \frac{1}{2(p - p_j)}, \quad p \to p_j\\
	& = - \frac{1}{2p}, \quad p \to p_n
\end{eqnarray*}
The volume of the W-P metric is
\begin{eqnarray*}
	\int \omega_{WP} &= \int 2\pi  \frac{\partial \gamma_{1}}{\partial \bar{p}}  d^{2} p =  2\pi \frac{1}{2i} \left[ \oint_{C_{n}} \gamma_{1} dp - \oint_{j} \gamma_{1} dp \right]\\
	& =  2\pi^{2} \left( - \frac{1}{2} + \frac{1}{2} (n-2)\right) =   \pi^{2} (n-3)
\end{eqnarray*}
For general $n$, this would be related to the total phase under adiabatic transport of a single cusp while keeping all others fixed. 

\bigskip
{\it Acknowledgements} We are pleased to thank A. Abanov,  Y.-H. Chiu, S. Ganeshan, A. Gromov, S. Grushevsky, S. Klevtsov, M. Laskin, S. Sun, L. Takhtajan, A. Waldron, P. Wiegmann, for discussions. Special thanks to A. Abanov, S. Ganeshan, S. Klevtsov, and P. Wiegmann for very helpful comments on an early draft.

\section*{References}
\bibliographystyle{unsrt}

\bibliography{FQHcusp.bib}



\end{document}